\documentclass[pdftex, a4paper, 11pt]{article}

\usepackage{amsmath, amsthm, amssymb}          
\usepackage[english]{babel}
\usepackage[latin1]{inputenc}
\usepackage{latexsym}
\usepackage{dsfont}
\usepackage{graphicx}
\usepackage{calc}
\usepackage{url}
\usepackage[nomessages]{fp} %
\usepackage[shortlabels,inline]{enumitem}

\usepackage{mathtools}

\usepackage{tikz}
\usetikzlibrary{patterns}
\usetikzlibrary{calc}
\usetikzlibrary{fadings}

\definecolor{darkblue}{rgb}{0,0,0.38}
\definecolor{darkred}{rgb}{0.6,0,0}
\definecolor{darkgreen}{rgb}{0.1,0.35,0}

\usepackage{hyperref}
\usepackage{xfrac}

\hypersetup{colorlinks, linkcolor=darkblue, citecolor=darkgreen, urlcolor=darkblue}

\usepackage[textsize=tiny]{todonotes}

\newcommand{\aR}{\mathbb{R}}
\newcommand{\nnR}{\mathbb{R}_{\geq 0}}
\newcommand{\dcup}{ \mathbin{\dot{\cup}} }
\newcommand{\sdiff}{ \mathbin{\triangle} }
\DeclareMathOperator{\conv}{conv}
\DeclareMathOperator{\cone}{cone}

\DeclareMathOperator{\E}{\mathbb{E}}

\renewcommand{\epsilon}{\varepsilon}

\newcommand{\OPT}{\mathsf{OPT}}

\newtheoremstyle{inproof}%
  {}%
  {}%
  {\itshape}%
  {}%
  {\texttt}%
  {.}%
  { }%
  {}%

\newtheorem{thm}{Theorem}
\newtheorem{lem}[thm]{Lemma}
\newtheorem{cor}[thm]{Corollary}
\newtheorem{defn}[thm]{Definition}

\theoremstyle{inproof}

\newtheorem*{claimInProof*}{Claim}

\usepackage[ruled,vlined,linesnumbered]{algorithm2e}
\SetKwInput{KwData}{Input}
\SetKwInput{KwResult}{Output}
\SetEndCharOfAlgoLine{}
\SetAlCapNameFnt{\sffamily}
\newcommand{\algoname}[1]{\textnormal{\textsf{#1}}}

\usepackage{geometry}
\usepackage[margin=0.5cm, font=small, labelfont=bf, labelsep=period ]{caption}

\title{Submodular Maximization through the Lens of Linear Programming}

\author{
Simon Bruggmann\thanks{
Department of Mathematics, ETH Zurich, Switzerland.
Email: \href{mailto:simon.bruggmann@ifor.math.ethz.ch}%
{simon.bruggmann@ifor.math.ethz.ch}.
Supported by the Swiss National Science Foundation grant
200021\_165866.}
\and
Rico Zenklusen\thanks{
Department of Mathematics, ETH Zurich, Switzerland.
Email: \href{mailto:ricoz@math.ethz.ch}%
{ricoz@math.ethz.ch}.
Supported by the Swiss National Science Foundation grant
200021\_165866.
}
}

\date{}

\begin{document}

\maketitle%

\begin{abstract}
	The simplex algorithm for linear programming is based on the fact that any local optimum with respect to the polyhedral neighborhood is also a global optimum. We show that a similar result carries over to submodular maximization. In particular, every local optimum of a constrained monotone submodular maximization problem yields a $1/2$-approximation, and we also present an appropriate extension to the non-monotone setting. 
However, reaching a local optimum quickly is a non-trivial task.

Moreover, we describe a fast and very general local search procedure that applies to a wide range of constraint families, and unifies as well as extends previous methods. In our framework, we match known approximation guarantees while disentangling and simplifying previous approaches. Moreover, despite its generality, we are able to show that our local search procedure is slightly faster than previous specialized methods. 

Furthermore, we resolve an open question on the relation between linear optimization and submodular maximization; namely, whether a linear optimization oracle may be enough to obtain strong approximation algorithms for submodular maximization. We show that this is not the case by providing an example of a constraint family on a ground set of size $n$ for which, if only given a linear optimization oracle, any algorithm for submodular maximization with a polynomial number of calls to the linear optimization oracle will have an approximation ratio of only $O ( \frac{1}{\sqrt{n}} \cdot \frac{\log  n}{\log\log n} )$.
\end{abstract}

\section{Introduction}
In this work, we consider the problem of maximizing a non-negative submodular function $f$ on a ground set $E$, subject to a membership constraint with respect to a given family $\mathcal{F} \subseteq 2^E$, i.e.,
\begin{equation*}
\max_{F\in \mathcal{F}} f(F)\enspace.\tag{$\mathsf{CSFM}$}
\end{equation*}
This is the most general form of the (non-negative) constrained submodular function maximization problem, which we abbreviate by $\mathsf{CSFM}$, as highlighted above.

A function $f \colon 2^E \to \aR$ is  \emph{submodular} if for all $S,T \subseteq E$ we have  $f(S\cup T) + f(S\cap T) \leq f(S) + f(T)$. 
Such a function $f$ is called \emph{monotone} if $f(S) \leq f(T)$ for all $S \subseteq T \subseteq E$. %
Throughout, we only consider non-negative submodular functions $f$ and we assume that they are given by a value oracle that for any set $S \subseteq E$ returns the value $f(S)$.

Submodular functions arise naturally in a wide variety of settings, examples being cut functions in weighted directed or undirected graphs, 
 coverage functions, 
 and rank functions of matroids (see, e.g.,~\cite{schrijver_2003_combinatorial}).
Moreover, the above definition of submodularity is equivalent to requiring the function $f$ to satisfy $f( T \cup \{x\}) - f(T) \leq f( S \cup \{x\}) - f(S)$ for all $S \subseteq T \subseteq E$ and $x \in E \setminus T$.
Submodular functions thus capture the property of diminishing marginal returns; a phenomenon that is observed in numerous relevant maximization problems, for example, in economics, game theory, or machine learning~\cite{LehmannLehmannNisan2006,BalcanBlumMansour2008,HartlineMirrokniSundararajan2008,CalinescuChekuriPalVondrak2011,MirzasoleimanKarbasiSarkarKrause2013,WeiIyerBilmes2014}.
Due to these connections, submodular function maximization problems have attracted considerable interest during recent years.

\subsection{Background}

The   problem of unconstrained submodular function maximization ($\mathsf{USFM}$) %
 generalizes many $\mathsf{NP}$-hard problems.
For example, the famous $\mathsf{APX}$-hard Maximum Cut problem, which is also one of Karp's classical $\mathsf{NP}$-hard problems~\cite{Karp1972}, is a special case of $\mathsf{USFM}$.
This implies $\mathsf{APX}$-hardness of $\mathsf{USFM}$, and thus also of $\mathsf{CSFM}$.
Research in this area has therefore primarily focused on approximation algorithms. 

In 2012, Buchbinder, Feldman, Naor, and Schwartz~\cite{BuchbinderFeldmanNaorSchwartz2012,buchbinder_2015_tight}  presented an efficient randomized algorithm with an approximation guarantee of $1/2$ for $\mathsf{USFM}$. %
This is optimal in the value oracle model since  Feige, Mirrokni, and Vondr\'ak~\cite{FeigeMirrokniVondrak2007,feige_2011_maximizing} showed that there is no algorithm with an  approximation factor better than $1/2$ that uses a subexponential number of value queries (even if restricted to symmetric submodular functions).
This hardness is information-theoretic and thus it does not rely on complexity-theoretic assumptions (such as $\mathsf{P} \neq \mathsf{NP}$).
In 2016, Buchbinder and Feldman~\cite{BuchbinderFeldman2016} showed that the algorithm in~\cite{buchbinder_2015_tight} can be derandomized, thus obtaining a deterministic $1/2$-approximation for $\mathsf{USFM}$.

The study of $\mathsf{CSFM}$ started in the 70's with the work of Fisher, Nemhauser, and Wolsey~\cite{NemhauserWolseyFisher1978,FisherNemhauserWolsey1978,NemhauserWolsey1978}. Among other results, they proved that for monotone $\mathsf{CSFM}$ under a single cardinality constraint, a greedy algorithm yields a $(1 - 1/e)$-approximation. %
For the same problem, they also showed that, for any $\epsilon > 0$, an exponential number of oracle calls is needed to obtain a $(1 - 1/e + \epsilon)$-approximation. 
Again, this hardness is purely information-theoretic. %
Later, algorithms with an optimal $(1-1/e)$-approximation guarantee have been obtained for monotone $\mathsf{CSFM}$ under a single knapsack constraint by Sviridenko~\cite{Sviridenko2004} and a single matroid constraint by Vondr\'ak~\cite{Vondrak2008} (see also~\cite{CalinescuChekuriPalVondrak2011}); moreover, a nearly optimal $(1-1/e-\epsilon)$-approximation (for any constant $\epsilon >0$) was first developed for the intersection of constantly many knapsack constraints by Kulik, Shachnai, and Tamir~\cite{kulik_2013_approximations}, and subsequently for the intersection of a single matroid constraint with constantly many knapsack constraints by Chekuri, Vondr\'ak, and Zenklusen~\cite{chekuri_2009_dependent,ChekuriVondrakZenklusen2010}.

Most current algorithms for $\mathsf{CSFM}$ are either greedy algorithms~\cite{NemhauserWolseyFisher1978,FisherNemhauserWolsey1978,NemhauserWolsey1978, ConfortiCornuejols1984, Sviridenko2004, GuptaRothSchoenebeckTalwar2010, FeldmanHarshawKarbasi2017}, 
local search methods~\cite{LeeMirrokniNagarajanSviridenko2009,lee_2010_maximizing,LeeSviridenkoVondrak2010,feige_2011_maximizing,FeldmanNaorSchwartzWard2011,Ward2012}, 
or relaxation-and-rounding procedures~\cite{ AgeevSviridenko2004, ChekuriVondrakZenklusen2010, FeldmanNaorSchwartz2011a, FeldmanNaorSchwartz2011b, CalinescuChekuriPalVondrak2011, ChekuriVondrakZenklusen2011, chekuri_2014_submodular}. 
Especially the last category comprises algorithms that can be used for a wide range of different constraint types and combinations thereof, instead of being highly tailored to a specific constraint family.
Almost all relaxation-and-rounding approaches in $\mathsf{CSFM}$ make use of the multilinear extension
 $F\colon [0,1]^E \rightarrow \mathbb{R}_{\geq 0}$ of the given submodular function $f$.
 For a point $x\in [0,1]^E$, it is defined by the expected value $F(x) \coloneqq \E[f(R(x))]$, where $R(x)$ is a random subset of $E$ containing each element $e\in E$ independently with probability $x(e)$. %
One main reason for using the multilinear extension $F$ of $f$ is that it can be approximately maximized (within a constant factor) over any down-closed solvable polytope $P$, i.e., a polytope over which one can efficiently solve linear programs. 
 This was shown for monotone $f$ by Vondr\'ak~\cite{Vondrak2008} and C{\u{a}}linescu, Chekuri, P\'al, and Vondr\'ak~\cite{CalinescuChekuriPalVondrak2011}, and for non-monotone $f$ (with a weaker approximation guarantee) first by Chekuri, Vondr\'ak, and Zenklusen~\cite{chekuri_2014_submodular}.
 Later, stronger constant factors were obtained through elegant techniques by Feldman, Naor, and Schwarz~\cite{FeldmanNaorSchwartz2011b}, Ene and Nguy$\tilde{\hat{\mathrm{e}}}$n~\cite{ene_2016_constrained}, 
  and Buchbinder and Feldman~\cite{buchbinder_2016_constrained}. In these results, the requirement of $P$ being solvable is only used in a black box fashion (by assuming that a linear optimization oracle is given for $P$).
This created hope that there may be a deeper connection between the existence of a linear optimization oracle and the possibility to obtain strong approximations for submodular maximization problems.
However, it is important to note that a good approximation for maximizing the multilinear extension $F$ is, in general, not enough to get a good approximation for maximizing the submodular function $f$.
This is because the former will usually yield a fractional point that has to be rounded to an integral point afterwards.
While this can be done without too much loss for some types of constraints,  e.g., by using contention resolution schemes (see~\cite{chekuri_2014_submodular}), the multilinear extension has unbounded integrality gap for other types of constraints.

Another strong connection between linear and submodular maximization surfaces in the matroid secretary problem. 
More precisely, Feldman and Zenklusen~\cite{FeldmanZenklusen2015} showed that any constant-competitive algorithm for the linear matroid secretary problem, even restricted to a particular matroid class, can be used in a black box manner to get a constant-competitive algorithm for the submodular matroid secretary problem over the same matroid class.

\subsection{Main results and techniques}
\label{subsec:mainResults}

The primary goal of this work is twofold. We want to obtain a better understanding of the relation between linear and submodular maximization, and present a unifying view on local search procedures, which not only encompasses known results but also allows for extensions.
In particular, we investigate whether an analogue to the simplex algorithm may also work to obtain good approximations for submodular maximization. Moreover, we show that access to a linear optimization oracle alone, for a given constraint set, is not sufficient to obtain strong approximation guarantees for constrained submodular maximization.

The main part of this work is about local search in the context of submodular maximization.
Motivated by the simplex algorithm, we consider local search from a geometric point of view.
This enables us to generalize previously proposed local search procedures and to simplify their analyses.
In fact, we show that local search always yields good approximations for submodular maximization  as long as the notion of neighborhood used satisfies certain geometric properties, which are fulfilled when, like in the simplex algorithm, we use the polyhedral adjacency structure to define neighborhoods.
Using this general viewpoint, we are able to match currently best approximation ratios for $\mathsf{CSFM}$ over $k$-intersection systems~\cite{lee_2010_maximizing,LeeSviridenkoVondrak2010}, which are the intersection of $k$ matroids on the same ground set, and $k$-exchange systems~\cite{FeldmanNaorSchwartzWard2011,Feldman2013} (see Appendix~\ref{append:kExchange} for a formal definition of $k$-exchange systems) in a unifying framework, and with even a slightly lower running time bound. Moreover, our framework extends beyond these settings.

Our local search results are based on the following key geometric statement, which we will refine later.
It provides a sufficient condition under which one can derive good approximation guarantees for a feasible solution $S$. We will reach conditions of this type with our local search techniques. To build up intuition, one can think of $A_1, \ldots , A_k$ as neighbors of $S$ with respect to some neighborhood structure, like the polyhedral neighborhood %
of some feasibility family $\mathcal{F}$,%
\footnote{The polyhedral neighborhood of $\mathcal{F} \subseteq 2^E$ is defined through its \emph{combinatorial polytope}
$P_{\mathcal{F}} \coloneqq \conv\big(\{\chi^F \mid F\in \mathcal{F}\}\big)$, where $\chi^A\in \{0,1\}^E$ for $A\subseteq E$ is the \emph{characteristic vector} of $A$. A set $F\in \mathcal{F}$ is a \emph{(polyhedral) neighbor} of $S\in \mathcal{F}$ if $\chi^F$ and $\chi^S$ are two adjacent vertices of $P_{\mathcal{F}}$ or if $F = S$.}
and of $T$ as an optimal solution $\OPT$ to the problem we consider.
\begin{thm}\label{thm:intro-main}
	Let $f \colon 2^E \to \nnR$ be a non-negative submodular function on some ground set $E$.
	If $S$, $T$, and $A_1, \ldots, A_k$ are subsets of $E$ such that
	\begin{align*}
	\chi^T \in
	  \Big( \chi^S + \cone\big(\{ \chi^{A_i}-\chi^S \mid i \in [k]  \}\big)     \Big)  \enspace,
	\end{align*}
	then there exist coefficients $\lambda_i \geq 0$ for $i \in [k]$ such that at most $|E|$ many of them are non-zero and  such that
	\begin{align*}
	 2 \cdot f(S) + \sum_{i=1}^k \lambda_i \big(f(A_i) - f(S)\big) \geq f(S \cup T) + f(S \cap T)   \enspace.
	\end{align*}
\end{thm}

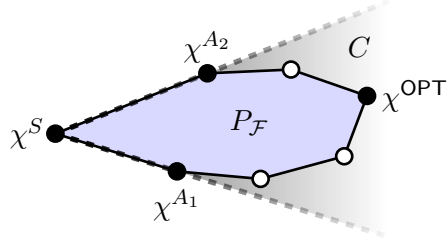
\begin{figure}
	\centering
\begin{tikzpicture}
	[every path/.append style={line width=1}]

	\pgfdeclarelayer{fg}
	\pgfdeclarelayer{bg}
	\pgfsetlayers{bg,main,fg}

	\coordinate (S) at (0,0);
	\coordinate (A1) at (1.6,-0.5);
	\coordinate (A2) at (2,0.8);
	\coordinate (B1) at (2.7,-0.6);
	\coordinate (B2) at (3.8,-0.3);
	\coordinate (B3) at (3.1,0.85);
	\coordinate (OPT) at (4.1, 0.5);	
	
	\coordinate (AA1) at ($ ($(S)$) + ($1.25*2.2*($ (A1)-(S) $)$)$);
	\coordinate (AA2) at ($ ($(S)$) + ($2.2*($ (A2)-(S) $)$)$);

\begin{pgfonlayer}{fg}
	\tikzstyle{coneline}=[dashed, path fading=east, black, line width=2pt]
	\draw[coneline] (S) to (AA1);
	\draw[coneline] (S) to (AA2);
\end{pgfonlayer}

	\draw[draw=none, top color=black!60, bottom color=white , shade, shading angle=90] ($(AA1)$) -- ($(S)$) -- ($(AA2)$);
	
	\fill[blue!15] (S) -- (A1) -- (B1) -- (B2) -- (OPT) -- (B3) -- (A2) -- (S);
	\draw (S) -- (A1) -- (B1) -- (B2) -- (OPT) -- (B3) -- (A2) -- (S);
	
	\node[left] at (S) {$\chi^S$};
	\node[below=0.1] at (A1) {$\chi^{A_1}$};	
	\node[above] at (A2) {$\chi^{A_2}$};	
	\node at (B1) {};
	\node at (B2) {};	
	\node at (B3) {};	
	\node[right=0.05] at (OPT) {$\chi^{\OPT}$};	
	
	\node at (2.55,0.15) {$P_\mathcal{F}$};	
	\node at (4,1.15) {$C$};	
	
	\begin{scope}[every node/.style={draw,shape=circle,minimum size=6,inner sep=0, line width=1}]

	\tikzstyle{full}=[fill=black]
	\tikzstyle{hollow}=[fill=white]

	\node[full] at (S) {};
	\node[full] at (A1) {};	
	\node[full] at (A2) {};
	\node[hollow] at (B1) {};
	\node[hollow] at (B2) {};	
	\node[hollow] at (B3) {};	
	\node[full] at (OPT) {};		
	
	\end{scope}
	
\end{tikzpicture}	
 	\caption{The vertex $\chi^S$ has the polyhedral neighbors $\chi^{A_1}$, $\chi^{A_2}$, and itself. The shifted cone $C = \chi^S + \cone(\{ \chi^{A_1} - \chi^S, \chi^{A_2} - \chi^S, \chi^S - \chi^S \})$ thus contains the whole polytope $P_\mathcal{F}$, and therefore also $\chi^{\OPT}$. 
	This basic property clearly holds for any polytope, even though we only need it for $\{0,1\}$-polytopes. For better illustration, the above graphic exemplifies the property on a polytope that is not $\{0,1\}$.}
	\label{fig:intro-cone}
\end{figure}

To illustrate how this result can be used, consider a \emph{monotone} non-negative submodular function $f \colon 2^E \rightarrow \mathbb{R}_{\geq 0}$ that we want to maximize over some feasibility family $\mathcal{F} \subseteq 2^E$, and let $P_{\mathcal{F}}=\conv(\{\chi^F \mid F\in \mathcal{F} \})$ be the convex hull of all characteristic vectors of $\mathcal{F}$, i.e., the combinatorial polytope corresponding to $\mathcal{F}$.
Assume that $S \in \mathcal{F}$ is a local optimum with respect to the polyhedral neighborhood of $\mathcal{F}$. We claim that the above theorem implies that $S$ is a $1/2$-approximation to the problem $\max_{F\in \mathcal{F}} f(F)$.
To see this, let $A_1,\ldots, A_k\in \mathcal{F}$ be the neighboring sets of $S$ and let $T=\OPT \in \mathcal{F}$ be an optimal solution. Now, the shifted cone
\begin{equation*}
C \coloneqq \chi^S + \cone\big(\{\chi^{A_i} - \chi^S \mid i\in [k]\}\big)
\end{equation*}
is the polyhedron defined by all the constraints in a facet-description of $P_{\mathcal{F}}$ that are tight at $\chi^S$. Hence, it holds that $P_{\mathcal{F}}\subseteq C$ (see Figure~\ref{fig:intro-cone} for an illustration).
In particular, this implies that $\chi^{\OPT}\in P_{\mathcal{F}}\subseteq C$, and thus, by Theorem~\ref{thm:intro-main}, there exist coefficients $\lambda _i \geq 0$ for $i \in [k]$ such that
\begin{equation}\label{eq:introExampleMonotone}
2 \cdot f(S) + \sum_{i=1}^k \lambda_i (f(A_i) - f(S))
  \geq f(S\cup \OPT) + f(S\cap \OPT) \geq f(\OPT)\enspace .
\end{equation}
Above, the second inequality follows by non-negativity and monotonicity of $f$, which allows for dropping the term $f(S\cap \OPT)\geq 0$ and using $f(S\cup \OPT) \geq f(\OPT)$, respectively.
Finally, the assumption that $S$ is a local optimum with respect to the polyhedral neighborhood of $\mathcal{F}$ ensures that $f(A_i) \leq f(S)$ for all $i\in [k]$, thus implying by~\eqref{eq:introExampleMonotone} that $S$ is a $1/2$-approximation.
Hence, whereas in linear programming any locally optimal set with respect to the polyhedral neighborhood is globally optimal, we obtain the following counterpart for monotone submodular maximization.
\begin{cor}\label{cor:localOptPolyhedron}
For the problem of constrained monotone  submodular maximization, any locally optimal set with respect to the polyhedral neighborhood is a $1/2$-approximation.
\end{cor}

In particular, Corollary~\ref{cor:localOptPolyhedron} implies one of the well-known results of Fisher, Nemhauser, and Wolsey~\cite{FisherNemhauserWolsey1978} which says that for monotone $\mathsf{CSFM}$ over a matroid constraint, any local optimum with respect to the polyhedral neighborhood is a $1/2$-approximation.
While the proof in~\cite{FisherNemhauserWolsey1978} makes use of exchange properties of matroids, we infer from the above that there is in fact no need to rely on such combinatorial structures. Indeed, Corollary~\ref{cor:localOptPolyhedron} shows that the statement is true for arbitrary feasibility families $\mathcal{F}$.

In addition, a slight generalization of Theorem~\ref{thm:intro-main} enables us to provide efficient local search procedures for the monotone \emph{and non-monotone} case for a large class of settings. This also allows for replicating existing results in~\cite{lee_2010_maximizing,LeeSviridenkoVondrak2010, FeldmanNaorSchwartzWard2011,Feldman2013} in a unified framework.

While the above results further increase the hope to obtain a strong link between linear and submodular optimization, we show that being able to optimize linear functions over some constraint family does not suffice to get good approximations for submodular maximization over the same constraints.
\begin{thm}\label{thm:intro-hardness}
	There exists no algorithm for constrained monotone submodular maximization that uses only a polynomial number of linear optimization oracle calls and achieves an approximation ratio of $\omega ( \frac{1}{\sqrt{n}} \cdot \frac{\log  n}{\log\log n} )$, where $n$ is the size of the ground set. This result also holds for randomized algorithms, where, as usual, the approximation ratio only has to be achieved in expectation.
\end{thm}
Theorem~\ref{thm:intro-hardness}  provides a strong separation between the difficulty of optimizing a linear function  and approximately maximizing a monotone submodular function over a given set of constraints.
In the example we construct to prove the above theorem, the family of feasible sets is down-closed and the submodular function $f \colon 2^E\rightarrow \mathbb{R}_{\geq 0}$ is even given explicitly.
The difficulty in our example therefore stems primarily from the family of feasible sets and the fact that we can only learn about feasible sets through the linear optimization oracle.

\subsection{Organization of the paper}

We start with a brief section on preliminaries, where we describe the so-called Lov\'asz extension of a submodular function. This extension plays a key role in our derivations.
In Section~\ref{sec:LS}, we show a generalization of Theorem~\ref{thm:intro-main}, and describe our local search procedures and implications thereof. 
Finally, Theorem~\ref{thm:intro-hardness} is proved in Section~\ref{sec:Hardness}.

\section{Preliminaries}

We first provide some basics on the Lov\'asz extension of a submodular function, which is a crucial tool in our approach. 
We introduce the extension for any set function without assuming submodularity. Let $f\colon2^E \rightarrow \mathbb{R}$  be a set function on a finite ground set $E$  of cardinality $n \coloneqq |E|$. Then, the Lov\'asz extension $f_{\mathsf{L}}$ of $f$ is the function $f_{\mathsf{L}}\colon [0,1]^E \rightarrow \mathbb{R}$ given by the expectation
\begin{equation*}
f_{\mathsf{L}}(x) \coloneqq \E[f(\{e\in E \mid x(e) \geq Z\})]\enspace,
\end{equation*}
where $x \in [0,1]^E$ and $Z$ is a uniform random variable within $[0,1]$. Equivalently, if we define $\{x_1, \ldots, x_n\} = \{x(e) \mid e\in E\}$ to be the different entries of $x$ ordered increasingly, i.e., $0\leq x_1 \leq x_2 \leq \ldots \leq x_n \leq 1$, then
\begin{equation}\label{eq:lovaszExtSlices}
f_{\mathsf{L}}(x) \coloneqq \sum_{i=1}^{n+1} (x_i - x_{i-1})\cdot f(\{e\in E \mid x(e) \geq x_i\})\enspace,
\end{equation}
where we set $x_0=0$ and $x_{n+1}=1$.
Clearly,~\eqref{eq:lovaszExtSlices} shows that the Lov\'asz extension of a function can be efficiently evaluated.

Moreover, its significance in submodular optimization stems from the fact that it captures the convex closure of a submodular function. More precisely, for any set function $f\colon 2^E\rightarrow \mathbb{R}$, its convex closure $f^-$ is defined as 
\begin{equation}\label{eq:convexClosure}
f^-(x) \coloneqq \min \Bigg\{\sum_{F\subseteq E} \alpha_F\cdot f(F)
 \;\Bigg\vert\;
 \alpha_F \geq 0 \text{ for } F\subseteq E,
  \sum_{F\subseteq E}\alpha_F = 1,
  \sum_{F\subseteq E} \alpha_F \cdot \chi^F = x
  \Bigg\}\enspace,
\end{equation}
 where $x\in [0,1]^E$. 
The convex closure is the point-wise largest convex extension of $f$. 
A crucial property which is heavily exploited in submodular minimization is that the Lov\'asz extension $f_{\mathsf{L}}$ is the same as the convex closure $f^-$ if and only if $f$ is submodular. This implies that for submodular functions, the convex closure can be evaluated efficiently. Moreover, it also shows that for submodular $f$, the Lov\'asz extension $f_{\mathsf{L}}$ is a convex function. This property is at the heart of many efficient submodular function minimization algorithms that use convex minimization techniques to minimize $f_{\mathsf{L}}$.

While this connection was mostly (though not exclusively) employed for submodular minimization problems, we too will regularly use the Lov\'asz extension and the above-highlighted two different ways to look at it, given by~\eqref{eq:lovaszExtSlices} and~\eqref{eq:convexClosure}, to derive inequalities that are crucial in our work. 
Moreover, the way we employ the Lov\'asz extension also allows for providing simpler and more elegant proofs for known inequalities that have been used in the context of submodular maximization.

\section{Generalized local search}\label{sec:LS}

In this section, we start by discussing geometric results underlying our local search approaches, and in particular prove (a generalization of) Theorem~\ref{thm:intro-main}. We then show how to derive local search procedures out of it and expand on applications of the suggested procedures. Compared to previously described local search algorithms, the analysis of our method is arguably simpler. It separates the geometric properties from submodularity, and therefore makes the underlying structure more visible.
In addition, our analysis also leads to a slightly better running time.

\subsection{Geometric results and good neighborhoods}

Throughout this section, $E$ is a finite ground set of cardinality $n\coloneqq |E|$, $\mathcal{F}\subseteq 2^E$ is a non-empty family of what we call \emph{feasible} sets, and $f\colon 2^E\rightarrow \mathbb{R}_{\geq 0}$ is a non-negative submodular function. We are interested in the $\mathsf{CSFM}$ problem $\max_{F\in \mathcal{F}}f(F)$. As usual, $f$ is assumed to be given through a value oracle, which, for any set $F\subseteq E$, returns the value $f(F)$.
There are different ways how the family of feasible sets $\mathcal{F}$ can be given.
Usually, we assume that we get a membership oracle for $\mathcal{F}$ which tells us for a set $F \subseteq E$ whether it is a member of $\mathcal{F}$ or not;
we will be explicit whenever we use different assumptions.

We start by proving the following generalization of Theorem~\ref{thm:intro-main}, which lies at the heart of our local search results.
\begin{thm}\label{thm:cone-result}%
\footnote{%
One can actually prove a slightly stronger version of this theorem, parameterized by two parameters $\alpha \geq 1$ and $0\leq \beta \leq \alpha-1$, where the main condition is replaced by
$\frac{1}{\alpha} \cdot (\chi^T + \beta \cdot \chi^{S\cap T}) \in \big(\chi^S + \cone(\{\chi^{A_i}-\chi^S \mid i\in [k]\})\big)$, and the implications~\eqref{eq:cone-result-i} and~\eqref{eq:cone-result-ii} become
$\chi^T + \beta \cdot \chi^{S\cap T} = \alpha \cdot \chi^S + \sum_{i=1}^k \lambda_i (\chi^{A_i} - \chi^S)$ and $(\alpha+1) \cdot f(S) + \sum_{i=1}^k \lambda_i (f(A_i)-f(S)) \geq f(S\cup T) + (\beta+1) \cdot f(S\cap T)$, respectively. Clearly, Theorem~\ref{thm:cone-result} corresponds to $\beta=\alpha-1$. However, since we do not need this more general statement, we restrict ourselves to the case $\beta=\alpha -1$.%
}
	Let $f \colon 2^E \to \nnR$ be a submodular function on some ground set $E$, and let $\alpha \geq 1$.
	If $S$, $T$, and $A_1, \ldots, A_k$ are subsets of $E$ such that
	\begin{align*}
		\frac{1}{\alpha} \cdot \big(  \chi^T + (\alpha -1) \cdot \chi^{S \cap T} \big)  \in
		\Big( \chi^S + \cone\big(\{ \chi^{A_i}-\chi^S \mid i \in [k]  \}\big)     \Big)  \enspace,
	\end{align*}
	 then there exist coefficients $\lambda_i \geq 0$ for $i \in [k]$ such that at most $|E|$ many of them are non-zero and such that
	 \begin{gather}
	 \chi^T + (\alpha -1) \cdot \chi^{S \cap T}
     = \alpha \cdot \chi^S + \sum_{i=1}^k \lambda_i  (\chi^{A_i} - \chi^S)\label{eq:cone-result-i}\enspace, \text{ and }\\
	(\alpha+1) \cdot f(S) + \sum_{i=1}^k \lambda_i \big(f(A_i) - f(S)\big)
     \geq f(S \cup T) + \alpha \cdot f(S \cap T) \label{eq:cone-result-ii}  \enspace.
	\end{gather}
\end{thm}
\begin{proof}
	By the assumption and Carath\'eodory's theorem for cones, there is a subset $J \subseteq [k]$ of size $|J| \leq |E|$ such that 
	$\frac{1}{\alpha} \cdot (  \chi^T + (\alpha - 1) \cdot \chi^{S \cap T} )  \in
	( \chi^S + \cone(\{ \chi^{A_i}-\chi^S \mid i \in J \})   )$.  
	Without loss of generality we therefore assume %
	 $k \leq |E|$ (otherwise, simply set $\lambda_i = 0$ for $i \in [k]\setminus J$).
	Hence, there exist coefficients $\lambda_i \geq 0$ for $i \in [k]$ such that 
	\begin{align*}
		\chi^T + (\alpha - 1) \cdot \chi^{S \cap T} = \alpha \cdot \chi^S + \sum_{i=1}^k \lambda_i  (\chi^{A_i} - \chi^S)   \enspace. 
	\end{align*}
	Letting $\lambda = \sum_{i=1}^k \lambda_i$, we can rewrite the above as
	\begin{align}\label{eq:coneAiST}
	 \sum_{i=1}^k \lambda_i  \chi^{A_i} 	= \chi^T + (\alpha - 1) \cdot \chi^{S \cap T} + (\lambda - \alpha) \cdot \chi^S  \enspace. 
	\end{align}
	 Note that if $\lambda = 0$, then $S=T$ and the theorem is trivially true.  Thus, we assume $\lambda >0$. 
   We now divide both sides of~\eqref{eq:coneAiST} by $\lambda$ and apply the Lov\'asz extension $f_{\mathsf{L}}$ of $f$ to both sides. We thus obtain
	 \begin{align}\label{eq:sumAi-fL}
		   \sum_{i=1}^k \frac{\lambda_i}{\lambda}  f(A_i) \geq
     f_{\mathsf{L}}\bigg(\sum_{i=1}^k \frac{\lambda_i}{\lambda} \chi^{A_i}\bigg)
		 = f_{\mathsf{L}}\Big(    \frac{1}{\lambda} \cdot (  \chi^T + (\alpha - 1) \cdot \chi^{S \cap T} + (\lambda - \alpha) \cdot \chi^S  )  \Big)  \enspace,
	 \end{align}
where the first inequality follows from the fact that the Lov\'asz extension is equal to the convex closure $f^-$ for submodular $f$ and from the definition of the convex closure (see~\eqref{eq:convexClosure}).

We now focus on developing the right-hand side of~\eqref{eq:sumAi-fL}. For brevity, let $x \coloneqq  \frac{1}{\lambda} \cdot (  \chi^T + (\alpha - 1) \cdot \chi^{S \cap T} + (\lambda - \alpha ) \cdot \chi^S  ) $. We claim that %
	\begin{align}\label{eq:fL-ST}
		f_{\mathsf{L}}(x) \geq \frac{1}{\lambda} \cdot \Big( f(S \cup T)   +  (\lambda - \alpha -1) \cdot f(S)  + \alpha \cdot f(S \cap T)   \Big) \enspace.
	\end{align}
	Notice that this indeed implies the desired result, because by putting \eqref{eq:sumAi-fL} and \eqref{eq:fL-ST} together, we get 
	\begin{align*}
		(\alpha +1) \cdot f(S) + \sum_{i=1}^k \lambda_i  \big(f(A_i) - f(S)\big) \geq f(S \cup T) + \alpha \cdot f(S \cap T)   \enspace.
	\end{align*}
	Hence, it remains to prove~\eqref{eq:fL-ST}. For this, we use the other viewpoint on the Lov\'asz extension given by~\eqref{eq:lovaszExtSlices}. Depending on the precise values of $\alpha$ and $\lambda$, the coordinates of $x$ have to be numbered differently to obtain an increasing sequence $0 \leq x_1 \leq \ldots  \leq x_n \leq 1$, as used in~\eqref{eq:lovaszExtSlices}. More precisely, three different orderings may be necessary, depending on the value of $\lambda - \alpha$. 
	We treat these orderings through the following case distinction:
	\begin{enumerate}[(i)]
		\item  $\lambda - \alpha \geq 1$:
		 In this case, it holds that 
		\begin{align*}
		f_{\mathsf{L}}(x) &= \frac{1}{\lambda} \cdot \Big( f(S \cup T)   + (\lambda - \alpha -1) \cdot f(S) + \alpha \cdot f(S \cap T) 
	   \Big) 
		 \enspace.
		\end{align*}
		
		\item $0 \leq \lambda - \alpha < 1$:
		In this case, we have that
		\begin{align*}
		f_{\mathsf{L}}(x) &= \frac{1}{\lambda} \cdot \Big( 
		(\lambda - \alpha ) \cdot f(S \cup T)   + (1-\lambda + \alpha ) \cdot f(T) + (\lambda - 1)\cdot f(S \cap T)     \Big) \\  
		&\geq 	\frac{1}{\lambda} \cdot \Big( 
		 f(S \cup T)   + (\lambda - \alpha -1 ) \cdot f(S) + \alpha \cdot f(S \cap T)   \Big)
		\enspace,
		\end{align*}
		where we used that $1 - \lambda + \alpha \geq 0$ and $f(T) \geq f(S \cup T) + f(S \cap T) -f(S)$ by submodularity of $f$.
		
		\item $\lambda - \alpha  <0$:
		In this case, it must hold that $S \subseteq T$; for otherwise, the vector defined on the right-hand side of~\eqref{eq:coneAiST} would have coordinates with strictly negative entries (more precisely, this would be the case for all coordinates corresponding to elements in $S\setminus T$), which is not possible since the left-hand side of~\eqref{eq:coneAiST} is clearly a non-negative vector.
		Hence, we can rewrite  \eqref{eq:coneAiST} as 
		\begin{align}\label{eq:coneAiSTsubset}
			\sum_{i=1}^k \lambda_i  \chi^{A_i} 	= \chi^T + (\lambda - 1) \cdot \chi^S
			= \chi^{T\setminus S} + \lambda  \cdot \chi^S   \enspace .
		\end{align}
		Moreover, if $S = T$, we have that 
		\begin{align*}
		f_{\mathsf{L}}(x) = \frac{1}{\lambda} \cdot \Big(
		\lambda  \cdot f( S)   
		    \Big)  = \frac{1}{\lambda} \cdot \Big( f(S \cup T)   + (\lambda - \alpha -1) \cdot f(S) + \alpha \cdot f(S \cap T) 
		\Big)
		\enspace.
		\end{align*}
		What remains is the case $S \subsetneq T$. Since the latter implies $\lambda \geq 1$, we obtain
		\begin{align*}
		f_{\mathsf{L}}(x) &= \frac{1}{\lambda} \cdot \Big(
		f( T)   + ( \lambda - 1) \cdot f(S) 
		   \Big) \\
		& = \frac{1}{\lambda} \cdot \Big( f(S \cup T)   + (\lambda - \alpha -1) \cdot f(S) + \alpha \cdot f(S \cap T) 
		\Big)
		\enspace.
		\end{align*}
	\end{enumerate}
Together, the three cases above prove \eqref{eq:fL-ST}, and thus complete the proof.
\end{proof}

To expand on how we exploit Theorem~\ref{thm:cone-result} to design local search procedures, we first formalize the notion of neighborhoods.
We then introduce the notion of an $\alpha$-conic neighborhood for $\mathcal{F}$, which describes what we need to apply Theorem~\ref{thm:cone-result}.
We highlight that our notion of neighborhood need not be symmetric.

\begin{defn}	
	A \emph{neighborhood function} $N$ for $\mathcal{F}$ is a function $N \colon \mathcal{F} \to 2^\mathcal{F}$ such that $S \in N(S)$ for every $S \in \mathcal{F}$. For $S \in \mathcal{F}$, the subfamily $N(S) \subseteq \mathcal{F}$ is called the \emph{neighborhood} of $S$, 
	and any set $T \in N(S)$ is said to be a \emph{neighbor} of $S$. 
Moreover, for $\alpha \geq 1$, $N$ is said to be \emph{$\alpha$-conic} for $\mathcal{F}$ if
	\begin{align*}
	\frac{1}{\alpha} \cdot \big(  \chi^T + (\alpha -1) \cdot \chi^{S \cap T} \big)  \in
	\Big( \chi^S + \cone\big(\{ \chi^A-\chi^S \mid A \in N(S)  \}\big)     \Big)  \qquad \forall S,T\in \mathcal{F} \enspace.
	\end{align*}
\end{defn}

A natural example of a neighborhood function, which originates from the combinatorial polytope $P_{\mathcal{F}} = \conv\big(\{ \chi^F \mid F \in \mathcal{F} \}\big)$ that is associated with $\mathcal{F}$, is the \emph{polyhedral neighborhood function} $N_{\mathcal{F}}$. It is defined as 
\begin{equation*}
N_\mathcal{F}(S) \coloneqq \{S\} \cup \{ T \in \mathcal{F} \mid \chi^S \text{ and } \chi^T \text{ are adjacent vertices of } P_\mathcal{F} \}\enspace \quad \forall S\in \mathcal{F}\enspace.
\end{equation*}
As already discussed in Section~\ref{subsec:mainResults}, it is well-known that the polyhedral neighborhood function is $1$-conic for any feasibility family $\mathcal{F}\subseteq 2^E$.

The following observation generalizes the motivating example discussed in the introduction, i.e., that for monotone $\mathsf{CSFM}$, local optima with respect to the polyhedral neighborhood are $1/2$-approximations.

\begin{thm}\label{thm:localOptimaAreGood}
For monotone $\mathsf{CSFM}$, any local optimum  with respect to an $\alpha$-conic neighborhood is a $1/(\alpha+1)$-approximation.
\end{thm}
\begin{proof}
Let $N$ be an $\alpha$-conic neighborhood function for the feasibility family $\mathcal{F}$.
Let $S \in \mathcal{F}$ be a local optimum and $T=\OPT \in \mathcal{F}$ be a global optimum to the considered $\mathsf{CSFM}$ problem. By invoking Theorem~\ref{thm:cone-result} with $A_1, \ldots, A_k$ being all sets in $N(S)$, we have that the condition of the theorem is satisfied due to $N$ being $\alpha$-conic. Moreover, if $f$ is the submodular function we want to maximize, the fact that $S$ is a local optimum implies $f(A_i) - f(S)\leq 0$ for every $i\in [k]$. Hence, by~\eqref{eq:cone-result-ii} we have as desired
\begin{equation*}
(\alpha+1)\cdot f(S) \geq f(S\cup \OPT) + \alpha\cdot f(S\cap \OPT)
    \geq f(S\cup \OPT) \geq f(\OPT)\enspace,
\end{equation*}
where the second and third inequality follow by non-negativity and monotonicity of $f$, respectively.
\end{proof}

While Theorem~\ref{thm:localOptimaAreGood} shows that local optima with respect to $\alpha$-conic neighborhoods are good in the monotone case, there are still two important questions, namely~\begin{enumerate*}[label=(\roman*)]
\item how to approach a local optimum quickly and
\item how to deal with non-monotone functions.
\end{enumerate*}

The following theorem addresses the first question by showing that as long as $(\alpha+1)\cdot f(S)$ is significantly smaller than $f(S\cup T) + \alpha \cdot f(S\cap T)$,  %
there is a neighbor of $S$ 
with function value significantly larger than that of $S$.
In Section~\ref{subsec:localSearchProc}, we will then present a fast local search procedure that addresses the second question.

\begin{thm}\label{thm:cone-increase}
	Let $f \colon 2^E \to \nnR$ be a submodular function on some ground set $E$, and let $\alpha \geq 1$.
	If $S$, $T$, and $A_1, \ldots, A_k$ are subsets of $E$ such that
	\begin{align*}
	\frac{1}{\alpha} \cdot \big(  \chi^T + (\alpha -1) \cdot \chi^{S \cap T} \big)  \in
	\Big( \chi^S + \cone\big(\{ \chi^{A_i}-\chi^S \mid i \in [k]  \}\big)     \Big)  \enspace,
	\end{align*}
	then at least one of the following two statements holds:
	\begin{enumerate}[\normalfont (i), itemsep=0em]
		\item\label{thm:cone-increase-i} $S$ and $T$ satisfy $ \displaystyle (\alpha + 1) \cdot f(S) \geq f( S \cup T) + \alpha \cdot f(S \cap T)    $.
		\item\label{thm:cone-increase-ii} There is $j \in [k]$ such that 
		$   f(A_j) - f(S) \geq \frac{1}{\alpha \cdot |E|} \cdot \big(  
		 f(S \cup T) + \alpha \cdot f(S \cap T) - (\alpha +1) \cdot f(S)     \big)   $.
	\end{enumerate}		
\end{thm}
\begin{proof}
	We assume that~\ref{thm:cone-increase-i} does not hold and show that this implies~\ref{thm:cone-increase-ii}.
	By Theorem~\ref{thm:cone-result}, we know that there exist coefficients $\lambda_i \geq 0$ for $i \in [k]$ such that at most $|E|$ many of them are non-zero and  such that
	\begin{gather}
	\chi^T + (\alpha -1) \cdot \chi^{S \cap T} = \alpha \cdot \chi^S + \sum_{i=1}^k \lambda_i  (\chi^{A_i} - \chi^S)\enspace,
	 \text{ and }\label{eq:cone-result-i-proof}\\
	(\alpha +1) \cdot f(S) + \sum_{i=1}^k \lambda_i \big(f(A_i) - f(S)\big) \geq f(S \cup T) + \alpha \cdot f(S \cap T)\label{eq:cone-result-ii-proof}  \enspace.
	\end{gather}
	From \eqref{eq:cone-result-ii-proof}, it immediately follows that
	\begin{align*}
	 \sum_{i=1}^k \lambda_i \big(f(A_i) - f(S)\big) \geq f(S \cup T) + \alpha \cdot f(S \cap T) - (\alpha +1) \cdot f(S) >0 \enspace,
	\end{align*}
	where the last strict inequality is implied by the assumption that~\ref{thm:cone-increase-i} does not hold.
	Since at most $|E|$ many of the coefficients $\lambda_i \geq 0$, $i \in [k]$, are non-zero, there must be $j \in [k]$ such that
	\begin{align*}
	\lambda_j \big(f(A_j) - f(S)\big) \geq \frac{1}{|E|} \cdot \Big( f(S \cup T) + \alpha \cdot f(S \cap T) - (\alpha +1) \cdot f(S) \Big) >0 \enspace.
	\end{align*}
	This implies
\begin{equation}\label{eq:AjGreaterThanS}
\lambda_j >0  \text{ and } f(A_j) > f(S)\enspace.
\end{equation}
	If $\lambda_j \leq \alpha $, we immediately get from the above that
	\begin{align*}
f(A_j) - f(S) \geq \frac{1}{\alpha \cdot |E|} \cdot \big(  
f(S \cup T) + \alpha \cdot f(S \cap T) - (\alpha +1) \cdot f(S)     \big)   \enspace,
	\end{align*}
	which shows~\ref{thm:cone-increase-ii}.
	We complete the proof by showing that $\lambda_j > \alpha$ is not possible. For the sake of contradiction, assume $\lambda_j > \alpha$. 
	Defining
	$\lambda = \sum_{i=1}^k \lambda_i$,
	we get by~\eqref{eq:cone-result-i-proof} that
	\begin{align}\label{eq:singleOutAj}
 \lambda_j \chi^{A_j} + \sum_{i \in [k], i \neq j} \lambda_i \chi^{A_i}  =
 	\sum_{i=1}^k \lambda_i \chi^{A_i}   =
 	\chi^T + (\alpha -1 ) \cdot \chi^{S \cap T} +  \underbrace{(\lambda - \alpha )}_{>0} {} \cdot \chi^S  \enspace.   
	\end{align}
	If $A_j \setminus S \neq \emptyset$, then $\lambda_j \leq 1$, which follows by considering a coordinate corresponding to any element $e\in A_j\setminus S$ in~\eqref{eq:singleOutAj}. However, because $\alpha \geq 1$ by assumption, this implies $\lambda_j \leq \alpha$, thus leading to the desired contradiction.
	Moreover,  $S \setminus A_j \neq \emptyset$ would imply $\lambda - \alpha \leq \lambda - \lambda_j$, which follows by considering a coordinate corresponding to any element $e\in S\setminus A_j$ in~\eqref{eq:singleOutAj}. Since the latter is equivalent to $\lambda_j \leq \alpha $, we get a contradiction again.	
	Last but not least, $S = A_j$ leads to  the contradiction $0 = f(A_j) - f(S) >0$, where the strict inequality comes from~\eqref{eq:AjGreaterThanS}.
	We therefore conclude that $\lambda_j \leq \alpha$ always holds, which finishes the proof.	
\end{proof}

\subsection{Local search procedures}\label{subsec:localSearchProc}

Using the theory developed above, we are now ready to describe our generalized local search procedures.
Again, let a ground set $E$, a family of feasible sets $\mathcal{F} \subseteq 2^E$, and a submodular function $f \colon 2^E \to \nnR$ be given.
Throughout the following, we denote by $\OPT \in \mathcal{F}$ an optimal solution of the problem $\max_{F \in \mathcal{F}} f(F)$.
Moreover, we assume that we have an $\alpha$-conic neighborhood function $N \colon \mathcal{F} \to 2^\mathcal{F}$ for $\mathcal{F}$, where $\alpha \geq 1$.

If the given function $f$ is monotone, Theorem~\ref{thm:cone-increase} immediately implies that the natural local search procedure that iteratively goes from a current solution $S\in \mathcal{F}$ to a set $A\in N(S)$ having highest objective value among all sets in $N(S)$ will converge quickly. We call such a step, which goes from some $S \in \mathcal{F}$ to the best set in $N(S)$, a \emph{most improving step}.
What is important to note is that 
\begin{align*}
	\max \{  f(S') \mid S' \in N(S)  \} \geq f(S)
\end{align*}
for any $S \in \mathcal{F}$ since our requirement on neighborhood functions ensures $S \in N(S)$.

\begin{thm}\label{thm:LS-monotone}
	Assume that the submodular function $f$ is monotone, and let $\epsilon > 0$.
	Starting with a set $S_0 \in \mathcal{F}$, %
	 we get after $\gamma \coloneqq  \min \big\{ |E| \cdot \lceil \log \big( \frac{ f( \OPT )}{\epsilon \cdot  f(S_0)}  \big) \rceil ,
	|E| \cdot \lceil \log \big( \frac{\alpha + 1 + \epsilon }{\epsilon}  \big) \rceil		\big\}		$
	 many most improving steps a set $S_\gamma \in \mathcal{F}$ with
	\begin{align*}
		(\alpha + 1 + \epsilon) \cdot f(S_\gamma) \geq f(\OPT)  \enspace .
	\end{align*}
\end{thm}
\begin{proof}
	Let $S_0, S_1, \ldots, S_\gamma \in \mathcal{F}$ be the sets we encounter while performing the $\gamma$ most improving steps.
	The way we defined a most improving step ensures that $f(S_i) \leq f(S_{i+1})$ for every $i \in \{ 0, 1, \ldots, \gamma -1 \}$. If $(\alpha+1)\cdot f(S_\gamma) \geq f(\OPT)$, then the statement clearly holds. Hence, assume $(\alpha+1)\cdot f(S_\gamma) < f(\OPT)$, which, by the fact that $f(S_i)$ is non-decreasing in $i$, implies $(\alpha+1)\cdot f(S_i) < f(\OPT)$ for all $i\in \{0,1,\ldots, \gamma\}$.
	 Since the neighborhood function $N$ is $\alpha$-conic, we can invoke Theorem~\ref{thm:cone-increase} with $T = \OPT \in \mathcal{F}$ and $\{A^i_j\}_j = N(S_i)$ to obtain for every $i \in \{0, 1, \ldots, \gamma-1\}$ that
	\begin{align*}
	f(S_{i+1}) - f(S_i) 
	&\geq \frac{1}{\alpha \cdot |E|} \cdot \big( f(S_i \cup \OPT ) + \alpha \cdot f(S_i \cap \OPT) - (\alpha + 1) \cdot f(S_i)  \big)  \\
	&\geq  \frac{1}{\alpha \cdot |E|}  \cdot \big( f(\OPT ) - (\alpha + 1) \cdot f(S_i)  \big)
	\enspace .
	\end{align*}
	Here, we used that $S_{i+1}$ was chosen as a best set in $N(S_i)$ and that $f$ is non-negative and monotone.
	From the above, we get that  $S_{i+1}$ satisfies
	\begin{align*}
	f(\OPT) - (\alpha + 1) \cdot f( S_{i+1} )   
	 &\leq   \Big(  1 - \frac{\alpha + 1}{\alpha \cdot |E|}       \Big)
	 \cdot \big( f(\OPT) - (\alpha + 1) \cdot f(S_i) \big)      \\
	 &\leq \Big(1 - \frac{1}{|E|} \Big) \cdot \big( f(\OPT) - (\alpha + 1) \cdot f(S_i) \big)      
	\enspace .
	\end{align*}
By performing a most improving step, we can therefore close at least a $1/|E|$-fraction of the gap between $f(\OPT)$ and $\alpha + 1$ times the function value of our current set.
Using the trivial bound $f(\OPT) - (\alpha + 1) \cdot f(S_0) \leq f(\OPT)$, we  conclude that
\begin{align}\label{eq:LS-monotone-aux}
\begin{split}
f(\OPT) - (\alpha + 1) \cdot f(S_\gamma) 
&\leq \Big(1- \frac{1}{|E|}\Big)^\gamma \cdot \big( f(\OPT) - (\alpha + 1) \cdot f(S_0) \big)  \\
& \leq \Big(1- \frac{1}{|E|}\Big)^\gamma \cdot f(\OPT) 
\enspace .
\end{split}
\end{align}
Hence, if $\gamma =  |E| \cdot \lceil \log \big( \frac{ f( \OPT )}{\epsilon \cdot  f(S_0)}  \big) \rceil 	$, we get from~\eqref{eq:LS-monotone-aux} that
\begin{align*}
f(\OPT) - (\alpha + 1) \cdot f(S_\gamma) 
 \leq \Big(1- \frac{1}{|E|}\Big)^\gamma \cdot f(\OPT)
 \leq \epsilon \cdot f(S_0)  
 \leq \epsilon \cdot f(S_\gamma)
\enspace ,
\end{align*}
as desired.
If, on the other hand, $\gamma =  |E| \cdot \lceil \log \big( \frac{\alpha + 1 + \epsilon }{\epsilon}  \big) \rceil	$, it then follows from~\eqref{eq:LS-monotone-aux} that
\begin{align*}
f(\OPT) - (\alpha + 1) \cdot f(S_\gamma) 
\leq \Big(1- \frac{1}{|E|}\Big)^\gamma \cdot f(\OPT)
\leq \frac{\epsilon}{\alpha + 1 + \epsilon} \cdot f(\OPT)
\enspace .
\end{align*}
Since this is equivalent to
\begin{align*}
\frac{\alpha + 1}{\alpha + 1 + \epsilon} \cdot f(\OPT) - \frac{\alpha + 1}{\alpha + 1 + \epsilon} \cdot (\alpha + 1 + \epsilon) \cdot f(S_\gamma) 
\leq 0
\enspace ,
\end{align*}
we conclude that
\begin{align*}
f(\OPT) \leq (\alpha + 1 + \epsilon) \cdot f(S_\gamma) 
\end{align*}
in this case as well.
\end{proof}

Our main goal, however, is to provide a local search procedure that works also for non-monotone submodular functions.
In order to do so, we will use an iterative approach that reduces to the problem
of finding a pair of sets $(S,Q)$ with $S, Q \in \mathcal{F}$ such that
\begin{align}\label{eq:LS-goal-eps}
	(\alpha + 1 + \epsilon) \cdot f(S) \geq f(Q\cup T) + \alpha \cdot f(Q \cap T)	\qquad \forall T\in \mathcal{F}\enspace.
\end{align}
Most previously proposed local search procedures try to find an (approximate) local maximizer $S \in \mathcal{F}$, show that $(S,S)$ satisfies~\eqref{eq:LS-goal-eps} (this also follows from Theorem~\ref{thm:cone-result}), and then continue from there.
Since~\eqref{eq:LS-goal-eps} is actually all we need, we do not insist on finding an (approximate) local maximizer but are already happy with $S,Q \in \mathcal{F}$ satisfying the weaker condition above.
In fact, this difference together with part~\ref{thm:cone-increase-ii} of Theorem~\ref{thm:cone-increase} is what enables us to slightly decrease the number of required steps compared to similar local search procedures.

Vondr\'ak~\cite{Vondrak2013} showed that exponentially many value oracle queries are needed to get a constant-factor approximation for maximizing non-monotone submodular functions over the set of bases of a given matroid.
Since no strong approximation factors are possible for non-monotone $\mathsf{CSFM}$ over general constraint families, non-monotone $\mathsf{CSFM}$ is typically studied for \emph{down-closed} feasibility families $\mathcal{F}\subseteq 2^E$, i.e., whenever $I\in \mathcal{F}$ and $J\subseteq I$, then also $J \in \mathcal{F}$. In the following, we therefore focus on this case, which also appears naturally in many relevant applications.

Whenever we have a down-closed feasibility family $\mathcal{F} \subseteq 2^E$, we assume that every singleton $\{u\}$ for $u \in E$ is feasible and fulfills $f(\{u\})> f(\emptyset)$; for otherwise, we could delete $u$ from $E$ without changing the problem (since $f$ is submodular).
Note that the above can easily be checked in linear time. 
When considering a down-closed family $\mathcal{F} \subseteq 2^E$ for $E \neq \emptyset$, a trivial starting set $S_0 \in \mathcal{F}$ (with $f(S_0) > 0$) for our local search procedure is thus given by the best singleton.

Given an $\alpha$-conic neighborhood function, where $\alpha \geq 1$, Algorithm~\ref{alg:DCBLS} describes a procedure to obtain a pair of feasible sets $(S,Q)$ fulfilling~\eqref{eq:LS-goal-eps}.
Starting with $S_0 \in \mathcal{F}$, it %
 constructs a sequence of feasible sets $S_0, S_1, \ldots, S_j$ by performing most improving steps.
For each $i=0, \ldots, j-1$, the algorithm also measures the progress $\Delta_i = f(S_{i+1}) - f(S_i)$ that is made by moving from $S_i$ to $S_{i+1}$.
If it happens that this progress becomes small, Theorem~\ref{thm:cone-increase} implies that we are done.
If, on the other hand, the progress the algorithm makes does not get considerably smaller for several iterations, then we conclude that Algorithm~\ref{alg:DCBLS} has made enough progress as well. 
This is true since always moving to a best neighbor ensures by Theorem~\ref{thm:cone-increase} that the algorithm is  closing at least a $1/|E|$-fraction of the gap $ f(S_i \cup T) + \alpha \cdot f(S_i \cap T) - (\alpha + 1) \cdot f(S_i) $  for an arbitrary $T \in \mathcal{F}$.
What complicates the task (compared to the monotone case) is that the value $f(S_i \cup T ) + \alpha \cdot f(S_i \cap T)$ also varies with varying $S_i$ (and it does so in a not necessarily monotone fashion).
Because of this, Algorithm~\ref{alg:DCBLS} returns $S, Q \in \mathcal{F}$ that satisfy \eqref{eq:LS-goal-eps}, and not a single set $S \in \mathcal{F}$ satisfying
\begin{align}\label{eq:propPrevLS}
(\alpha + 1 + \epsilon) \cdot f(S) \geq f(S \cup T) + \alpha \cdot f(S \cap T)	\qquad \forall T\in \mathcal{F}\enspace,
\end{align}
as it was typical for previous local search procedures, e.g., see~\cite{lee_2010_maximizing,LeeSviridenkoVondrak2010,FeldmanNaorSchwartzWard2011}.
We will later show how one can design a fast local search algorithm for non-monotone $\mathsf{CSFM}$ that only requires pairs fulfilling~\eqref{eq:LS-goal-eps}, instead of the classical stronger condition~\eqref{eq:propPrevLS}.
As briefly mentioned above, this is the reason why our approach is faster than previous local search techniques, which, for an appropriate choice of $\alpha$ depending on the considered problem, had a bound of $O(\frac{|E|}{\epsilon} \cdot \log(\alpha \cdot |E|))$ on the number of iterations to find a set fulfilling~\eqref{eq:propPrevLS}. The following theorem highlights our speed improvement and shows that Algorithm~\ref{alg:DCBLS} finds a pair of sets $(S,Q)$ satisfying~\eqref{eq:LS-goal-eps} with a number of iterations depending only logarithmically on $1/\epsilon$.

\begin{algorithm}[t]
	\caption{Basic-Local-Search$( E, \mathcal{F}, \alpha,  N,f, \epsilon)$}\label{alg:DCBLS}
	
	\KwData{Ground set $E$,  down-closed family of feasible sets $\mathcal{F} \subseteq 2^E$, $\alpha$-conic neighborhood function $N \colon \mathcal{F} \to 2^\mathcal{F}$ for $\mathcal{F}$, where $\alpha \geq 1$, submodular function $f \colon 2^E \to \nnR$, error parameter $\epsilon > 0$.}
	\KwResult{Pair $(S,Q)$ with $S, Q \in \mathcal{F}$ and $(\alpha + 1 + \epsilon) \cdot f(S) \geq f(Q\cup T) + \alpha \cdot f(Q \cap T)$ for any $T \in \mathcal{F} $.}
	
	$S_0 \gets \arg\max \{ f(S') \mid    S' \subseteq E,\ |S'|=1 \}  $ \;
	$S_1 \gets \arg\max \{ f(S') \mid  S' \in N(S_0) \}  $ \;
	$\Delta_0 \gets f(S_1) - f(S_0)$ \;
	$i \gets 0$ \;
	$j \gets 0$ \;
	\While{$\Delta_j  >  \frac{\epsilon}{\alpha \cdot |E|}  \cdot f(S_j)$}{
		\If{$\Delta_j \leq \Delta_i /2$}{
			$i \gets j$ \;
		}
		\If{$j-i \geq 2\cdot |E|$}{
			\Return $(S_j, S_i)$ \;\label{algline:earlyTerm}
		}
		$j \gets j+1$ \;
		$S_{j+1} \gets \arg\max \{  f(S') \mid S' \in N(S_j)  \} $ \;
		$\Delta_j \gets f(S_{j+1}) - f(S_j)$ \;
	}
	\Return $(S_j,S_j)$ \;\label{algline:lateTerm}
	
\end{algorithm}
\begin{thm}\label{thm:result-DCBLS}
	Algorithm~\ref{alg:DCBLS} returns after
	\smash{$O\big(  |E| \cdot  \log \big( \frac{\alpha \cdot |E| }{\epsilon }  \big) \big) $}
	many iterations (executions of the while-loop)
	a pair of sets $(S,Q)$ with $S, Q \in \mathcal{F}$ such that
	\begin{align}\label{eq:result-DCBLS}
	(\alpha + 1 + \epsilon) \cdot f(S) \geq f(Q\cup T) + \alpha \cdot f(Q \cap T)	\qquad \forall T \in \mathcal{F} \enspace .
	\end{align}
	In particular, if $f$ is monotone, it holds %
	 that 
	\begin{align}\label{eq:result-DCBLS-monotone}
	(\alpha + 1 + \epsilon) \cdot f(S) \geq f( T ) 	 \qquad \forall T \in \mathcal{F} \enspace .
	\end{align}
\end{thm}
The above theorem shows that we could also use Algorithm~\ref{alg:DCBLS} for monotone $\mathsf{CSFM}$. In that setting, however, the algorithm described in Theorem~\ref{thm:LS-monotone} is arguably simpler.
\begin{proof}
	We first note that because $f$ is non-negative,~\eqref{eq:result-DCBLS-monotone} is an easy consequence of~\eqref{eq:result-DCBLS} if $f$ is moreover monotone.
	
	We now prove that Algorithm~\ref{alg:DCBLS} works correctly. To this end, assume that for a given input the algorithm terminates and returns the pair $(S,Q)$.
	Since only feasible sets are considered throughout the algorithm, it is clear that $S,Q \in \mathcal{F}$.
	We now fix an arbitrary $T \in \mathcal{F}$ and show that 
	\begin{align}\label{eq:LS-goal-eps-proof} 
	(\alpha + 1 + \epsilon) \cdot f(S) \geq f(Q\cup T) + \alpha \cdot f(Q \cap T)	 \enspace .
	\end{align}

There are two ways for Algorithm~\ref{alg:DCBLS} to terminate. We start by considering the case that it terminates at line~\ref{algline:lateTerm}, i.e., we have $\Delta_j  \leq  \frac{\epsilon}{\alpha \cdot |E|}  \cdot f(S_j)$ for some $j \geq 0$, meaning that $(S,Q) = (S_j, S_j)$.
	We moreover assume that  $(\alpha + 1) \cdot f(S_j) < f( S_j \cup T) + \alpha \cdot f(S_j \cap T) $, since otherwise \eqref{eq:LS-goal-eps-proof} is trivially true. Theorem~\ref{thm:cone-increase} then implies that
	\begin{align*}
		\Delta_j = f(S_{j+1}) - f(S_j) \geq \frac{1}{\alpha \cdot |E|} \cdot \big(  
		f(S_j \cup T) + \alpha \cdot f(S_j \cap T) - (\alpha +1) \cdot f(S_j)     \big) \enspace.
	\end{align*}
	Together with the assumption that $\Delta_j  \leq  \frac{\epsilon}{\alpha \cdot |E|}  \cdot f(S_j)$, this yields \eqref{eq:LS-goal-eps-proof} for $(S,Q) = (S_j, S_j)$.

	Next, assume that Algorithm~\ref{alg:DCBLS} terminates at line~\ref{algline:earlyTerm}, i.e., $j-i \geq 2 \cdot |E| $ for some $i,j \geq 0$, which implies that the returned pair is $(S,Q) = (S_j, S_i)$.
	In this case, we have that $\Delta_i  >  \frac{\epsilon}{\alpha \cdot |E|}  \cdot f(S_i) \geq 0$ (since the algorithm did not terminate earlier) and 	 $\Delta_\ell > \Delta_i / 2 >0$ for every $\ell = i, \ldots, j$.
	Moreover, Theorem~\ref{thm:cone-increase} %
	 yields that by the choice of $S_{i+1}$, we have
	\begin{align*}
	\Delta_i = f(S_{i+1}) - f(S_i) \geq \frac{1}{\alpha \cdot |E|} \cdot \big(  
	f(S_i \cup T) + \alpha \cdot f(S_i \cap T) - (\alpha +1) \cdot f(S_i)     \big) \enspace,
	\end{align*}
which clearly holds if~\ref{thm:cone-increase-ii} of Theorem~\ref{thm:cone-increase} applies, and also if~\ref{thm:cone-increase-i} applies because $\Delta_i >0$.
	Hence, we get 
	\begin{align*}
		f(S_j) - f(S_i)   =  \sum_{\ell=i}^{j-1} \Delta_\ell  \geq \underbrace{(j-i)}_{\geq 2\cdot |E|} \cdot \frac{\Delta_i}{2}
		\geq  \frac{1}{\alpha} \cdot \big(  
		f(S_i \cup T) + \alpha \cdot f(S_i \cap T) - (\alpha +1) \cdot f(S_i)     \big)  \enspace .
	\end{align*}
	Using the fact that $f(S_i) \leq f(S_{i+1}) \leq \ldots \leq f(S_j)$, which holds because the sets $S_\ell$, $\ell = i+1 , \ldots, j$, are obtained by performing most improving steps, this implies 
	\begin{align*}
	(\alpha + 1) \cdot f(S_j) \geq \alpha \cdot f(S_j) + f(S_i)   \geq   
	f(S_i \cup T) + \alpha \cdot f(S_i \cap T)  \enspace ,
	\end{align*}
	as desired.
	
	In order to prove the bound on the number of iterations which Algorithm~\ref{alg:DCBLS} performs, and therefore also that the algorithm actually terminates, we let $0=i_0 < i_1 < i_2 <  \ldots $ be the different values that the variable $i$ in the algorithm takes.
	It is easy to see that whenever the variable $i$ %
	 is set to a new value, there are at most $2 \cdot |E|$ many subsequent iterations in which the value of $i$ remains unchanged, because after precisely $2\cdot |E|$ subsequent iterations without any change in $i$, the algorithm terminates at line~\ref{algline:earlyTerm}.
	
	Letting $\gamma \coloneqq  \lceil \frac{1}{\log 2} \cdot \log \big( \frac{\alpha \cdot |E| \cdot f(\OPT)}{\epsilon \cdot  f(S_0)}  \big) \rceil $, we now claim that $i_\gamma$ never gets defined, %
	meaning that
\begin{equation}\label{eq:upperBoundItDCBLS}
2 \cdot |E| \cdot \gamma +1 = O\bigg(  |E| \cdot  \log \Big( \frac{\alpha \cdot |E| \cdot f(\OPT)}{\epsilon \cdot  f(S_0)}  \Big) \bigg)
\end{equation}
is an upper bound on the number of iterations that Algorithm~\ref{alg:DCBLS} performs.
	Indeed, assume for the sake of contradiction that $i_\gamma$ gets defined.
	We first observe that $\Delta_{i_{\ell}} \leq  \Delta_{i_{\ell-1}} /2 $ for each $\ell = 1, \ldots, \gamma$.
	Moreover,  it trivially holds that $\Delta_{i_0} = \Delta_0 = f(S_1) - f(S_0) \leq f(\OPT)$. 
	Hence, it follows that 
	\begin{align}\label{eq:bound-DeltaIAlpha}
		\Delta_{i_{\gamma}} \leq \Big( \frac{1}{2} \Big)^\gamma \cdot \Delta_{i_0} \leq \frac{\epsilon}{\alpha \cdot |E|} \cdot f(S_0) \leq \frac{\epsilon}{\alpha\cdot |E|} \cdot f(S_{i_{\gamma}}) \enspace ,
	\end{align}
which contradicts the assumption that the algorithm did not terminate %
before the variable $i$ was set to $i_\gamma$. %
Hence,~\eqref{eq:upperBoundItDCBLS} is indeed an upper bound on the number of iterations performed by Algorithm~\ref{alg:DCBLS}. To prove the bound on the number of iterations claimed by Theorem~\ref{thm:result-DCBLS}, it suffices to observe that $\frac{f(\OPT)}{f(S_0)} \leq |E|$. This holds because non-negativity and submodularity of $f$ together with the fact that $S_0$ is the best singleton imply that for any set $F\subseteq E$, we have $f(F) \leq \sum_{u\in F}f(\{u\}) \leq |E|\cdot f(S_0)$. %
\end{proof}

When we want to maximize non-monotone submodular functions, a single pair of sets that satisfies~\eqref{eq:LS-goal-eps}  is not enough (neither is a single local maximizer).
To get around this problem, we use the idea of~\cite{lee_2010_maximizing,Feldman2013} and iteratively call our basic procedure, i.e., Algorithm~\ref{alg:DCBLS}, to build up a sequence of pairs $(S_i, Q_i)$ on successively smaller ground sets.

Algorithm~\ref{alg:DCILS} formally describes our procedure for non-monotone submodular functions. %
In its formulation, the following additional notation is used.
For a family of feasible sets $\mathcal{F} \subseteq 2^E$, a subset $E' \subseteq E$, and a neighborhood function $N \colon \mathcal{F} \to 2^\mathcal{F}$, we denote by $N_{E'}$ the restricted neighborhood function $N_{E'} \colon \mathcal{F}\cap 2^{E'} \to 2^{\mathcal{F}\cap 2^{E'}}$ defined by $N_{E'}(S) \coloneqq \{ A \in N(S) \mid A \subseteq E' \}$ for each $S \in \mathcal{F}\cap2^{E'}$.
Note that if  $N $ is an $\alpha$-conic neighborhood function for $\mathcal{F}$, %
then $N_{E'}$ is an $\alpha$-conic neighborhood function for $\mathcal{F}\cap 2^{E'}$.
Moreover, by definition it holds for any $S \in \mathcal{F}\cap 2^{E'}$ that the size of the restricted neighborhood  $N_{E'}(S)$ is at most the size of $N(S)$.
Hence, reducing the size of the ground set does not result in an increase in the size of any neighborhood.

\begin{algorithm}[t]
	\caption{Iterative-Local-Search$( E, \mathcal{F}, \alpha, N, f, \epsilon )$}\label{alg:DCILS}
	
	\KwData{Ground set $E$,  down-closed family of feasible sets $\mathcal{F} \subseteq 2^E$,  $\alpha$-conic neighborhood function $N \colon \mathcal{F} \to 2^\mathcal{F}$ for $\mathcal{F}$, where $\alpha \geq 1$, submodular function $f \colon 2^E \to \nnR$, error parameter $\epsilon > 0$.}%
	\KwResult{Set $S \in \mathcal{F}$ with $(\lfloor\alpha\rfloor + 1) \cdot (\alpha + 1 + \epsilon) \cdot f(S) \geq \lfloor\alpha\rfloor \cdot f(T) $ for any $T \in \mathcal{F} $.}
	
	$i \gets 1$ \;
	$E_i \gets E$ \;
	\While{$i \leq \lfloor\alpha\rfloor + 1$}{
		\uIf{$E_i \neq \emptyset$}{
		$(S_i,Q_i) \gets \algoname{Basic-Local-Search$( E_i, \mathcal{F}\cap 2^{E_i}, \alpha,  N_{E_i},f|_{E_i} , \epsilon)$}$ \label{algline:callToDCBLS}\;
		}
		\Else{
		$(S_i,Q_i) \gets (\emptyset, \emptyset)$ \;
		}
		$E_{i+1} \gets E_i \setminus Q_i$ \;
		$ i \gets i + 1$ \;
	}
	$S \gets \arg\max \{  f(S') \mid S' \in \{S_1, \ldots, S_{\lfloor\alpha\rfloor + 1} \} \}$  \;
	\Return $S$
	
\end{algorithm}

The following theorem shows that Algorithm~\ref{alg:DCILS} works for non-monotone $\mathsf{CSFM}$.
\begin{thm}\label{thm:result-DCILS}
Algorithm~\ref{alg:DCILS} returns a set $S\in \mathcal{F}$ satisfying
	\begin{align}\label{eq:result-DCILS}
	(\lfloor\alpha\rfloor + 1) \cdot (\alpha + 1 + \epsilon) \cdot f(S) \geq \lfloor\alpha\rfloor \cdot f( T)	\qquad \forall T \in \mathcal{F}\enspace.
	\end{align}
Moreover, the total number of most improving steps performed by Algorithm~\ref{alg:DCILS} (through the calls to \algoname{Basic-Local-Search}) is bounded by \smash{$O\big( \alpha \cdot  |E| \cdot  \log \big( \frac{\alpha \cdot |E| }{\epsilon }  \big) \big) $}.
\end{thm}

\begin{proof}
	We first note that the bound on the number of iterations immediately follows from Theorem~\ref{thm:result-DCBLS} since Algorithm~\ref{alg:DCILS} calls \algoname{Basic-Local-Search} at most $\lfloor \alpha \rfloor + 1$ times. 
	Moreover, it is clear that the set $S$ returned by the algorithm is in $\mathcal{F}$.
	
	Fixing an arbitrary $T \in \mathcal{F}$, we show that $S$ satisfies~\eqref{eq:result-DCILS}. 
	To do so, we let $(S_i, Q_i)$, $i = 1, \ldots, \lfloor\alpha\rfloor + 1$, be the pairs of sets constructed by the algorithm in the while-loop.
		Defining the sets $E_i$ as in the algorithm, it holds that $S_i, Q_i \subseteq E_i$ for every $i = 1, \ldots, \lfloor\alpha\rfloor + 1$.
	Moreover, the modified ground sets $E_i$ form a chain, i.e., we have that $E_{\lfloor\alpha\rfloor + 1} \subseteq E_{\lfloor\alpha\rfloor} \subseteq \ldots \subseteq E_2 \subseteq E_1 = E$. 
	For $i = 1, \ldots, \lfloor\alpha\rfloor+1$, we now let $T_i \coloneqq T \cap E_i$ and observe that since $\mathcal{F}$ is down-closed, we have $T_i \in \mathcal{F} \cap 2^{E_i}$.  Theorem~\ref{thm:result-DCBLS} thus implies 
	\begin{align}\label{eq:proof-DCILS-result-DCBLS}
		(\alpha + 1 + \epsilon) \cdot f(S) \geq (\alpha + 1 + \epsilon) \cdot f(S_i) \geq f(Q_i \cup T_i) + \alpha \cdot f(Q_i \cap T_i)
	\end{align}
	for every $i = 1, \ldots, \lfloor\alpha\rfloor + 1$ (this trivially holds if $E_i = \emptyset$ since $f$ is non-negative).
	Adding the $\lfloor\alpha\rfloor + 1$ instances of inequality~\eqref{eq:proof-DCILS-result-DCBLS} yields
	\begin{align}\label{eq:proof-DCILS-sum}
		(\lfloor\alpha\rfloor + 1) \cdot (\alpha + 1 + \epsilon) \cdot f(S) \geq
		\sum_{i=1}^{\lfloor\alpha\rfloor + 1} f(Q_i \cup T_i) + \alpha \cdot \sum_{i=1}^{\lfloor\alpha\rfloor + 1} f(Q_i \cap T_i) \enspace .
	\end{align}
We claim that
\begin{align}\label{eq:proof-DCILS-sumLovasz}
\sum_{i=1}^{\lfloor\alpha\rfloor + 1} f(Q_i \cup T_i) + \alpha \cdot \sum_{i=1}^{\lfloor\alpha\rfloor + 1} f(Q_i \cap T_i) 
\geq \lfloor\alpha\rfloor \cdot f(T)  \enspace ,
\end{align}
which, together with~\eqref{eq:proof-DCILS-sum}, finishes the proof of Theorem~\ref{thm:result-DCILS}.
One way to prove~\eqref{eq:proof-DCILS-sumLovasz} is to employ a reasoning presented in~\cite{Feldman2013}, which uses an inductive argument. 
What we do, though, is proving~\eqref{eq:proof-DCILS-sumLovasz} by exploiting the power of the Lov\'asz extension $f_{\mathsf{L}}$ of $f$ again.
We think that this makes the argument more concise and easier to follow, while again highlighting the strength of the Lov\'asz extension to prove inequalities in the context of submodular maximization.

Letting
\begin{equation*}
x \coloneqq \frac{1}{(\lfloor\alpha\rfloor +1) \cdot (\alpha + 1)} \cdot
 \Bigg( \sum_{i=1}^{\lfloor\alpha\rfloor + 1} \chi^{Q_i \cup T_i} + \alpha \cdot \sum_{i=1}^{\lfloor\alpha\rfloor + 1} \chi^{Q_i \cap T_i}  \Bigg) \in [0,1 ]^E \enspace,
\end{equation*}
we first observe that
\begin{alignat*}{3}
x(e) &\geq \frac{1}{\alpha+1} \qquad &&\forall e \in T\enspace, \text{ and }\\
x(e) &\leq \frac{1}{(\lfloor\alpha\rfloor+1) \cdot (\alpha + 1)} \qquad &&\forall e\in E\setminus T\enspace.
\end{alignat*}
Indeed, for any $e\in T$, either there is some $i\in [\lfloor\alpha\rfloor+1]$ such that $e\in Q_i$ or $e \notin Q_i$ for all $i\in [\lfloor\alpha\rfloor+1]$.
In the first case, the $i$-th term in each of the two sums in the definition of $x$ contains a $1$ in the coordinate corresponding to $e$, implying $x(e) \geq 1/(\lfloor\alpha\rfloor +1)  \geq 1/(\alpha+1)$, as desired.
In the second case, we have for any $i\in [\lfloor\alpha\rfloor+1]$ that $e\in E_i$, and thus also $e\in T_i$; this in turn means that each term $\chi^{Q_i\cup T_i}$ in the first sum in our definition of $x$ has a $1$ in the coordinate corresponding to $e$, thus again implying $x(e) \geq 1/(\alpha+1)$. 
Conversely, any element $e\in E\setminus T$ is not contained in any $T_i$ for $i\in[\lfloor\alpha\rfloor+1]$, and it is contained in $Q_i$ for at most one $i\in [\lfloor\alpha\rfloor+1]$ because the $Q_i$ are disjoint. This implies $x(e) \leq 1/((\lfloor\alpha\rfloor +1)\cdot(\alpha+1))$, as claimed.

Since $f$ is non-negative, we obtain by the definition of the Lov\'asz extension (see~\eqref{eq:lovaszExtSlices}) that
\begin{equation*}
f_{\mathsf{L}}(x) \geq  \bigg(\frac{1}{\alpha+1} - \frac{1}{(\lfloor\alpha\rfloor+1)\cdot(\alpha+1)}\bigg)\cdot f(T) = \frac{\lfloor\alpha\rfloor}{(\lfloor\alpha\rfloor+1)\cdot(\alpha + 1)} \cdot f(T)\enspace.
\end{equation*}
By using the viewpoint that the Lov\'asz extension of a submodular function is the same as its convex closure (see~\eqref{eq:convexClosure}), we have
\begin{equation*}
f_{\mathsf{L}}(x) \leq \frac{1}{(\lfloor\alpha\rfloor+1)\cdot(\alpha+1)} \cdot \Bigg(
\sum_{i=1}^{\lfloor\alpha\rfloor+1} f(Q_i \cup T_i) + \alpha\cdot
\sum_{i=1}^{\lfloor\alpha\rfloor+1} f(Q_i \cap T_i)
\Bigg)\enspace,
\end{equation*}
and~\eqref{eq:proof-DCILS-sumLovasz} follows by combining the above two inequalities.
\end{proof}

We remark that instead of repeating the while-loop of Algorithm~\ref{alg:DCILS} exactly $\lfloor\alpha\rfloor+1$ times, one could also stop at any iteration $\kappa \leq \lfloor\alpha\rfloor+1$. The above proof technique would then lead to a generalization of Theorem~\ref{thm:result-DCILS}, where the guarantee on $S$ becomes 
$\kappa \cdot (\alpha+1+\epsilon)\cdot f(S) \geq (\kappa-1)\cdot f(T)$ for any $T\in \mathcal{F}$, and the number of most improving steps is bounded by $O(\kappa\cdot |E| \cdot \log(\frac{\alpha\cdot |E|}{\epsilon}))$. However, we note that we cannot derive the same conclusion for $\kappa> \lfloor\alpha\rfloor+1$.
For all applications we consider later, we only need the statement provided by Theorem~\ref{thm:result-DCILS}. Because of this and for the sake of simplicity, we focused on this more specialized statement (which corresponds to $\kappa = \lfloor\alpha\rfloor +1$).

\subsection{Applications}\label{subsec:localSearchApp}

In a single step of our local search procedures, we go through the neighborhood of a given feasible set to find its best neighbor. Thus, in order to achieve polynomial running time, we need to use a neighborhood function $N$ for $\mathcal{F} \subseteq 2^E$ that has the property that the neighborhood $N(S)$ of any set $S \in \mathcal{F}$ can be constructed efficiently. In particular, it is necessary (but not sufficient) that the size of the neighborhood of any feasible set is polynomially bounded in the size of the input.
On the other hand, for our approach to yield good approximations, we need neighborhood functions that are $\alpha$-conic for $\alpha \geq 1$ as small as possible.
Unfortunately (though not surprisingly), these two goals usually conflict. %

Because the polyhedral neighborhood function $N_\mathcal{F}$ for $\mathcal{F}$ is always $1$-conic,
we obtain the following corollary through Theorems~\ref{thm:LS-monotone}  and \ref{thm:result-DCILS}.
\begin{cor}\label{cor:LS-polyhedralN}
Consider a $\mathsf{CSFM}$ problem with  feasibility family $\mathcal{F}$, and assume that for any $S\in \mathcal{F}$, one can efficiently construct $N_{\mathcal{F}}(S)$. Given a starting set $S_0 \in \mathcal{F}$ and fixing $\epsilon >0$, the algorithm described in Theorem~\ref{thm:LS-monotone}, respectively Algorithm~\ref{alg:DCILS}, with neighborhood function $N_{\mathcal{F}}$ allows for efficiently obtaining a solution that is a:
\begin{enumerate}[label=\normalfont(\roman*),nosep]
\item $\frac{1}{2+\epsilon}$-approximation if $f$ is monotone;
\item\label{item:LS-polyhedralNonM} $\frac{1}{4+\epsilon}$-approximation even for $f$ non-monotone if $\mathcal{F}$ is down-closed.
\end{enumerate}
\end{cor}
Notice that to obtain point~\ref{item:LS-polyhedralNonM} of the above corollary, the $\epsilon$ to be used in Theorem~\ref{thm:result-DCILS} (and therefore also in Algorithm~\ref{alg:DCILS}) is not the same as the one of the corollary; for example, invoking Theorem~\ref{thm:result-DCILS} with an error parameter of $\epsilon/2$ works out.
Also note that the subproblems we solve in line~\ref{algline:callToDCBLS} of Algorithm~\ref{alg:DCILS} use restrictions $(N_{\mathcal{F}})_{E_i}$ of the neighborhood function $N_{\mathcal{F}}$. Such a restriction $(N_{\mathcal{F}})_{E_i}$ corresponds to the polyhedral neighborhood function of the face of the polytope $P_{\mathcal{F}}$ where all coordinates of $E\setminus E_i$ are set to $0$. Since the polyhedral neighborhood of a face of a polytope is a subset of the polyhedral neighborhood of the original polytope, we can efficiently compute the neighborhood on any face if we can do so for the whole polytope. Hence, all subproblems encountered in line~\ref{algline:callToDCBLS} of Algorithm~\ref{alg:DCILS} can indeed be solved efficiently.

Corollary~\ref{cor:LS-polyhedralN} has several interesting implications. In particular, it shows that for monotone $\mathsf{CSFM}$ over a matroid constraint, it does not only hold that any local optimum is a $1/2$-approximation (see Corollary~\ref{cor:localOptPolyhedron} or~\cite{FisherNemhauserWolsey1978}), but also that we can exploit this fact algorithmically to get an efficient $1/(2+\epsilon)$-approximation.
For this, the only property about matroids that we need is that if $\mathcal{F}$ is equal to all independent sets of a matroid on a ground set of size $n$, then the corresponding polytope $P_{\mathcal{F}}$, known as the \emph{matroid polytope}, has only $O(n^2)$ different edge directions at any given vertex (see, e.g., Theorem 40.6 in volume B of~\cite{schrijver_2003_combinatorial}).
Moreover, all these edge directions can be constructed easily, which 
allows for computing the polyhedral neighborhood efficiently.
Hence, Corollary~\ref{cor:LS-polyhedralN} shows that there is no need to rely on deep combinatorial properties to prove that natural local search algorithms work for $\mathsf{CSFM}$.

Another implication of Corollary~\ref{cor:LS-polyhedralN} is that constant-factor approximations for $\mathsf{CSFM}$ can be obtained for any  $\mathcal{F}$ which is given through an inequality description of $P_\mathcal{F}$ and satisfies that $P_{\mathcal{F}}$ is non-degenerate. %
If the submodular function is non-monotone, then $\mathcal{F}$ needs furthermore to be down-closed.

To our knowledge, this %
 is the first result on local search for $\mathsf{CSFM}$ of such generality.
Moreover, it shows that, conceptually, the basic idea of the simplex algorithm also works for submodular maximization. In general, it is difficult to find good neighbors, though.

For many relevant families $\mathcal{F} \subseteq 2^E$, however, we cannot afford to search the whole polyhedral neighborhoods since their sizes may be exponential in the size of the ground set.
In these cases, the challenge is to come up with a good notion of neighborhood that can be constructed efficiently for every feasible set while still permitting good approximations. 
One neighborhood function that has a particularly simple form is defined below. It is inspired by~\cite{LeeSviridenkoVondrak2010} and~\cite{FeldmanNaorSchwartzWard2011}, and stems from the idea that two sets should be considered neighboring if they differ in few elements.
\begin{defn}
	Let $k,p  \geq 1$. For $\mathcal{F} \subseteq 2^E$, the \emph{$(k,p)$-swap neighborhood function} $N_p^k \colon \mathcal{F} \to 2^\mathcal{F}$ is defined by $N_p^k(S) \coloneqq \{ T \in \mathcal{F} \mid |T \setminus S| \leq p,\ |S\setminus T | \leq (k-1)p+1 \}$ for each $S \in \mathcal{F}$. Hence, $T \in \mathcal{F}$ is a neighbor of $S \in \mathcal{F}$ if and only if $T$ can be obtained from $S$ by adding at most $p$ %
	 and deleting at most $(k-1)p + 1$ elements.
\end{defn}
Note that if $k$ and $p$ are assumed to be constant, we can efficiently construct $N_p^k(S)$ for any $S \in \mathcal{F}$. 

The $(k,p)$-swap neighborhood turns out to be very useful when dealing with so-called $k$-intersection systems, which are the intersection of $k$ matroids on a common ground set, and $k$-exchange systems (the latter were introduced in~\cite{FeldmanNaorSchwartzWard2011} and we provide a formal definition in Appendix~\ref{append:kExchange}). 
Lee, Sviridenko, and Vondr\'ak~\cite{LeeSviridenkoVondrak2010} showed that for arbitrary but fixed $k \geq 2$ and $\epsilon >0$, there exists a polynomial local search algorithm for maximizing a non-negative submodular function $f$ over a $k$-intersection system $\mathcal{F}$ with approximation guarantee $\frac{1}{k+ \epsilon}$ for monotone $f$ and $\frac{k-1}{k^2 + \epsilon}$ for general $f$. 
More precisely, they do not explicitly state an algorithm for the non-monotone case but rely on an approach introduced by Lee, Mirrokni, Nagarajan, and Sviridenko~\cite{lee_2010_maximizing}, which, like our Algorithm~\ref{alg:DCILS}, makes repeated calls to a simpler ``basic'' procedure; %
\cite{LeeSviridenkoVondrak2010}~provides an improved basic procedure (for monotone functions) which gives the above approximation guarantee for non-monotone functions when using it within the high-level approach from~\cite{lee_2010_maximizing}.
Thereafter, Feldman, Naor, Schwartz, and Ward~\cite{FeldmanNaorSchwartzWard2011}, respectively Feldman~\cite{Feldman2013}, showed that the same results also hold for $k$-exchange systems. For the non-monotone case, Feldman~\cite{Feldman2013} iteratively uses the algorithm from~\cite{FeldmanNaorSchwartzWard2011} 
as basic procedure, again using the same high-level framework as~\cite{lee_2010_maximizing}.
For monotone $f$, Ward~\cite{Ward2012} later provided a different local search algorithm that attains an %
 approximation ratio of $ \frac{2}{k + 3 + \epsilon}$ for $k$-exchange systems, and, contrary to previous approaches, has a polynomial running time dependency on $1 / \epsilon$.
The techniques used in that procedure do not seem to generalize to the non-monotone case, however.

In the following, we show how our framework allows for obtaining the same approximation guarantees for $k$-intersection systems as in~\cite{LeeSviridenkoVondrak2010} and for $k$-exchange systems as in~\cite{FeldmanNaorSchwartzWard2011,Feldman2013} in a unifying way.\footnote{%
We have to consider both $k$-intersection systems and $k$-exchange systems since neither of these two classes of set systems contains the other (see~\cite{Feldman2013}).} 
While our approach achieves the same approximation guarantees, we provide a somewhat cleaner and disentangled analysis. In addition, our general local search procedure even requires %
fewer steps (we will expand on this later).
Again, all we need to do to apply our approach is showing that the $(k,p)$-swap neighborhood is $\alpha$-conic for the considered constraints (for some appropriately chosen $\alpha \geq 1$). This is the statement of the following theorem.
\begin{thm}\label{thm:LS-geom-ksystems}
	Let $\mathcal{F} \subseteq 2^E$ be a $k$-intersection system or a $k$-exchange system for $k \geq 2$. Then, for any $p\geq 1$, $N_p^k$ is a $(k-1+1/p)$-conic neighborhood function for $\mathcal{F}$.
\end{thm}
The proof of the above theorem, which roughly follows the first parts of the proofs in~\cite{LeeSviridenkoVondrak2010} and~\cite{FeldmanNaorSchwartzWard2011}, is deferred to Appendices~\ref{append:kIntersection} and~\ref{append:kExchange}, respectively.
Combining the above with Theorems~\ref{thm:LS-monotone} and~\ref{thm:result-DCILS}, we get the desired approximation algorithms.
\begin{cor}\label{cor:LS-kpN}
  Consider a $\mathsf{CSFM}$ problem over a feasibility family $\mathcal{F}$ that is either a $k$-intersection system or a $k$-exchange system for some fixed $k\geq 2$.
Then, for any fixed $\epsilon >0$ and $p=p(k,\epsilon)\geq 1$ chosen large enough, the algorithm described in Theorem~\ref{thm:LS-monotone}, respectively Algorithm~\ref{alg:DCILS}, with neighborhood function $N_p^k$ %
allows for efficiently obtaining a solution that is a:
\begin{enumerate}[label=\normalfont(\roman*),nosep]
\item $\frac{1}{k+\epsilon}$-approximation if $f$ is monotone;
\item $\frac{k-1}{k^2+\epsilon}$-approximation even for $f$ non-monotone. %
\end{enumerate}
\end{cor}
Note that for the algorithm of Theorem~\ref{thm:LS-monotone}, we can simply choose the best singleton as starting set $S_0 \in \mathcal{F}$ (we make our usual assumptions on the down-closed family $\mathcal{F}$).
Again, the $\epsilon$ which is used in our algorithms has to be chosen smaller than the one in the approximation guarantees of Corollary~\ref{cor:LS-kpN}. %
Since constructing the $(k,p)$-swap neighborhood of a feasible set requires $|E|^{O(k\cdot p)}$ time, the algorithms also have running times that depend exponentially on $k$. %
This is why $k \geq 2$ %
is assumed to be constant in Corollary~\ref{cor:LS-kpN}.

We now compare the running time of the algorithms in~\cite{LeeSviridenkoVondrak2010,FeldmanNaorSchwartzWard2011,Feldman2013} to the running time of our approach. 
Since our algorithm and the ones in~\cite{LeeSviridenkoVondrak2010,Feldman2013} for not necessarily monotone functions all make $k$ calls to basic procedures for monotone functions (in case of~\cite{Feldman2013},
 the algorithm from~\cite{FeldmanNaorSchwartzWard2011}; in our case, Algorithm~\ref{alg:DCBLS}), we only need to consider the running times of the employed basic procedures.%
\footnote{%
As stated before, Lee, Sviridenko, and Vondr\'ak~\cite{LeeSviridenkoVondrak2010} remark that when using their method as basic procedure in the iterative local search algorithm from Lee, Mirrokni, Nagarajan, and Sviridenko~\cite{lee_2010_maximizing}, then one gets the aforementioned approximation guarantee for the non-monotone case. The iterative algorithm from~\cite{lee_2010_maximizing}, as it is given there, makes $k+1$ calls to a basic procedure for monotone functions, but this number actually reduces to $k$ when using the basic procedure from~\cite{LeeSviridenkoVondrak2010}.}
The basic local search procedure for monotone functions from~\cite{FeldmanNaorSchwartzWard2011} requires \smash{$O\big( \frac{|E|}{\epsilon} \cdot \log(|E|)   \big) $} steps, where each step considers a $(k,p)$-swap neighborhood of a current set and moves to a neighboring set.
Moreover,  the same bound can be shown to hold for the basic procedure described %
in~\cite{LeeSviridenkoVondrak2010} when improving the parameters in their algorithm slightly.
From Theorem~\ref{thm:result-DCBLS}, however, we get that Algorithm~\ref{alg:DCBLS} requires only \smash{$O\big(|E| \cdot \log(\frac{|E|}{\epsilon}) \big)$} steps.
Since the worst-case time complexity of a single step and the dependency on parameters that are assumed to be constant is the same for all algorithms, we conclude that our algorithm  achieves a slightly better running time with no loss in the approximation guarantee and despite the fact that our algorithm and analysis consider a more general setting. 
In other words, the number of steps required by the algorithms in~\cite{LeeSviridenkoVondrak2010,FeldmanNaorSchwartzWard2011,Feldman2013} depends linearly on $1/\epsilon$ while this dependence is only logarithmic in our approach.

The class of $k$-intersection systems contains many interesting combinatorial optimization problems, such as Bipartite Matchings, Branchings in Digraphs (both for $k=2$), and, more generally, $k$-Dimensional Matchings (see~\cite{LeeSviridenkoVondrak2010}).
The same holds true for the class of $k$-exchange systems, since they generalize, for example, the problems of $b$-Matchings (for $k=2$), $k$-Set Packings, Asymmetric Traveling Salesman\footnote{%
Here, the family of feasible sets consists of all Hamiltonian cycles including all their subsets in a complete directed graph.
}
 (for $k=3$), and Independent Sets in $(k+1)$-Claw-Free Graphs (see~\cite{FeldmanNaorSchwartzWard2011}).
We remark that Hazan, Safra, and Schwartz~\cite{HazanSafraSchwartz2006} showed that it is $\mathsf{NP}$-hard to approximate both the Maximum $k$-Set Packing problem and the Maximum $k$-Dimensional Matching problem to within a factor of $O(\log k/k)$.
Hence, under the assumption that $\mathsf{P} \neq \mathsf{NP}$, our algorithms are at most a logarithmic factor away from being best possible.

We also hope that our approach will be helpful for obtaining further approximation algorithms for $\mathsf{CSFM}$ beyond the discussed settings.
For this, it suffices to find an efficiently computable neighborhood function that is $\alpha$-conic for the considered constraints, for a parameter $\alpha \geq 1$ as small as possible. Extending the approach to a new constraint family thus only requires a proof of a geometric property that is independent of the submodular function to be maximized.

\section{Hardness of approximation}\label{sec:Hardness}

In this section, we prove Theorem~\ref{thm:intro-hardness}, which %
states that a linear optimization oracle alone is not enough to find good approximations for submodular maximization with only polynomially many calls to the oracle.
We start by describing a family of monotone $\mathsf{CSFM}$ problems, parameterized by the size $n$ of the ground set, which we will then show to imply Theorem~\ref{thm:intro-hardness}.

\subsection{A family of bad instances}\label{sec:hardnessExample}

Let $E$ be a ground set of cardinality $n \coloneqq |E|$. 
To simplify notation, we assume that $n \geq 4
$ is the square of some positive number, i.e., $\sqrt{n} \in \mathbb{Z}$.
The set $E$ is partitioned into $\sqrt{n}$ many sets $S_1, \ldots, S_{\sqrt{n}}$, each of cardinality $\sqrt{n}$. 
Moreover, we set $\beta \coloneqq c \cdot \lceil \frac{\log n}{\log \log n}\rceil$, where $c \in \mathbb{Z}_{>0}$ is a constant to be determined later.
From each set $S_i$, $i \in [\sqrt{n}]$, we choose an element $t_i \in S_i$ independently and uniformly at random, and define $T\coloneqq \{t_1, \ldots, t_{\sqrt{n}} \}$ to be the set consisting precisely of all these elements. 
The family $\mathcal{F}$ of feasible subsets of $ E$ is then %
 defined as 
\begin{align*}
\mathcal{F} \coloneqq \big\{F \subseteq E \;\big\vert\; |\{ i \in [\sqrt{n}] \mid F \cap S_i \neq \emptyset \}| \leq \beta \big\} \cup \big\{F \subseteq E \;\big\vert\; F \subseteq T \big\}   \enspace .
\end{align*}
In particular, $\mathcal{F}$ is easily seen to be down-closed.
Moreover, the submodular function $f\colon 2^E \rightarrow \mathbb{Z}_{\geq 0}$ that we want to maximize over $\mathcal{F}$ is given by
\begin{align*}
f(F) \coloneqq |\{ i \in [\sqrt{n}] \mid F \cap S_i \neq \emptyset \}| \qquad \forall F \subseteq E
\enspace .
\end{align*}
Hence, for a subset $F \subseteq E$, $f(F)$ is equal to the number of parts $S_i$ which the set $F$ intersects.
Clearly, $f$ is a non-negative monotone submodular function; actually, $f$ is a coverage function. %

We define a partition $\mathcal{F} = \mathcal{S} \dcup \mathcal{T}$ of the feasible sets, where
\begin{align*}
\mathcal{S} &\coloneqq \{F\in \mathcal{F} \mid f(F) \leq \beta\} = \{F \subseteq E \mid |\{ i \in [\sqrt{n}] \mid F \cap S_i \neq \emptyset \}| \leq \beta \}\enspace,\text{ and }\\
\mathcal{T} &\coloneqq \{F \in \mathcal{F} \mid f(F) > \beta\} = \{F \subseteq E \mid F \subseteq T \text{ and } |F| >\beta\}\enspace.
\end{align*}
We refer to sets $F \in \mathcal{S}$ as \emph{standard feasible} sets  and to sets $F \in \mathcal{T}$ as \emph{special feasible} sets.
While $f(F) \leq \beta =  c \cdot \lceil\frac{\log n}{\log\log n}\rceil$ for every standard feasible set $F \in \mathcal{S}$, we have $f(F)  > \beta$ for every special feasible set $F\in \mathcal{T}$. 
In particular, since $f(T) = \sqrt{n}$, it holds for  $n$ large enough that the set $T$ (which we do not know) is the unique maximizer of $\max_{F \in \mathcal{F}} f(F)$. %
To prove the theorem, it thus suffices to show that a polynomial number of calls to a linear optimization oracle over $\mathcal{F}$ is not enough to find a special feasible set.
 Notice that if any special feasible set $F\in \mathcal{T}$ is found (which is easy to recognize), %
 then we can recover $T$ by calling the linear optimization oracle with the weights
\begin{align*}
w(s) \coloneqq \begin{cases}
n \enspace &\textup{if } s \in F \enspace, \\
1 \enspace &\textup{if } s \in E \setminus F \enspace.
\end{cases}
\end{align*}
Hence, the discovery of any special feasible set allows for retrieving $T$.

In order to prove Theorem~\ref{thm:intro-hardness}, we can even assume that an algorithm for the monotone $\mathsf{CSFM}$ problem $\max_{F\in \mathcal{F}}f(F)$ defined above has all the information about how the instance is constructed, except for knowing the elements $t_i$, $i \in [\sqrt{n}]$, of course.
However, we recall that the only access that it has to $\mathcal{F}$ is over an optimization oracle, which, for any $w\in \mathbb{R}^E$, returns an optimal solution to the linear optimization problem $\max_{F\in \mathcal{F}} w(F)$. 
To rule out that the linear optimization oracle returns a special feasible set ``by coincidence'', we make some very mild assumptions on what the oracle returns in case of ties.
Since these assumptions are only about how ties are broken, they still lead to a valid linear optimization oracle.
Assume that there is an arbitrary but fixed numbering of the elements in $E$ (that is independent of $T$), and that, in case of ties, the linear optimization oracle returns among all optimal sets of smallest cardinality the one that is lexicographically minimal with respect to this numbering.
If, however, there are both standard and special feasible sets which are optimal, we assume that the oracle disregards all special feasible sets and applies the above rule only to optimal standard feasible sets.
 In particular, this ensures that for any linear objective $w\in \mathbb{R}^E$, the linear optimization oracle with respect to $w$ returns the same set as when called with respect to $w^+$, where $w^+$ is defined by $w^+(e) \coloneqq \max\{w(e),0\}$ for each $e\in E$. Therefore, we can assume that any algorithm only makes calls to the linear optimization oracle with objectives within $\mathbb{R}^E_{\geq 0}$. 
 Moreover, for any linear objective $w \in \nnR^E$, we have $w(T) = \max_{V \in \mathcal{T}} w(V)$. Thus, the rule that no special feasible set is returned whenever there is an optimal standard feasible set ensures that calling the linear optimization oracle with $w \in \nnR^E$ yields a special feasible set only if $w(T) > \max_{U \in \mathcal{S}} w(U)$.

In the next section, we show how the above example implies Theorem~\ref{thm:intro-hardness}.

\subsection{Proof of Theorem~\ref{thm:intro-hardness}}

Let $\mathcal{A}$ be an algorithm for approximately solving $\mathsf{CSFM}$ given a linear optimization oracle and assume that $\mathcal{A}$ makes at most $O(n^d)$ many calls to the oracle, where $d>0$ is some constant. We prove Theorem~\ref{thm:intro-hardness} by showing that this algorithm has an approximation guarantee of at most $O ( \frac{1}{\sqrt{n}} \cdot \frac{\log  n}{\log\log n} )$ on the family of instances described in Section~\ref{sec:hardnessExample}, with $c$ being set to $c \coloneqq 2d+2$.
More precisely, we prove this statement by showing the following result which states that detecting a special feasible set with a single oracle call is very unlikely.
\begin{lem}\label{lem:singleOracleCallBad}
Let $w\in \mathbb{R}_{\geq 0}^E$ and let $T$ be a random subset of $E$ containing a uniformly random element of each set $S_i$ for $i\in [\sqrt{n}]$. Then, for large enough $n$, it holds that
\begin{equation*}
\Pr\Big[ w(T) > \max_{U\in \mathcal{S}} w(U)\Big] \leq n^{-d-1}\enspace,
\end{equation*}
where $\mathcal{S}$ is defined as before.
\end{lem}

Before we prove Lemma~\ref{lem:singleOracleCallBad}, we observe that it indeed implies Theorem~\ref{thm:intro-hardness}.

\begin{proof}[Proof of Theorem~\ref{thm:intro-hardness}]
We consider the monotone 	$\mathsf{CSFM}$ problem $\max_{F\in \mathcal{F}} f(F)$ described in Section~\ref{sec:hardnessExample}, where $n$ denotes the size of the ground set.
Using a union bound in combination with Lemma~\ref{lem:singleOracleCallBad}, we have that an algorithm $\mathcal{A}$ with $O(n^d)$ calls to the linear optimization oracle will detect a special feasible set with probability at most
\begin{equation*}
n^{-d-1} \cdot O(n^d)  = O(n^{-1})\enspace.
\end{equation*}
Therefore, if we denote by $Q_{\mathcal{A}} \in \mathcal{F}$ the output of algorithm $\mathcal{A}$, we have
\begin{align*}
\E[f(Q_{\mathcal{A}})] {} \leq {} & %
\Pr[\mathcal{A} \text{ does not encounter a special feasible set}]\cdot \beta
 \\
 & {}+{} \Pr[\mathcal{A} \text{ encounters a special feasible set}]\cdot f(T) \\
{} \leq {} &  \beta + \Pr[\mathcal{A} \text{ encounters a special feasible set}]\cdot \sqrt{n} \\
{} = {} & O(\beta + n^{-1} \cdot \sqrt{n} ) \\
{} = {} & O\Big(\frac{\log n}{\log\log n}\Big)\enspace.
\end{align*}
Because the optimal value of the considered $\mathsf{CSFM}$ problem is $f(T) = \sqrt{n}$, Theorem~\ref{thm:intro-hardness} follows.
\end{proof}

Hence, it only remains to prove Lemma~\ref{lem:singleOracleCallBad}.

\begin{proof}[Proof of Lemma~\ref{lem:singleOracleCallBad}]
We first show that we can assume that the weight of each set $S_i$, $i \in [\sqrt{n}]$, is either $0$ or $1$.
To this end, we start with a general $w\in \mathbb{R}^E_{\geq0}$ as stated in the lemma and define a vector $\overline{w}\in [0,1]^E \subseteq \mathbb{R}^E_{\geq 0}$ by setting, for any $i\in [\sqrt{n}]$ and $e\in S_i$,
\begin{equation*}
\overline{w}(e) \coloneqq
\begin{cases}
\frac{w(e)}{w(S_i)} &\textup{if } w(S_i) > 0 \enspace,\\
0 \;\;(= w(e)) &\textup{if } w(S_i) = 0 \enspace.
\end{cases}
\end{equation*}
Hence, it holds for every $i \in [\sqrt{n}]$ that
\begin{equation}\label{eq:relWToWBar}
w(S_i) \cdot \overline{w}(e)   = w(e) \qquad \forall %
e\in S_i \enspace,
\end{equation}
and
\begin{equation*}
\overline{w}(S_i) =
\begin{cases}
1 &\textup{if } w(S_i) > 0 \enspace,\\
0 &\textup{if } w(S_i) = 0 \enspace.
\end{cases}
\end{equation*}
The following claim implies that whenever the linear optimization oracle with objective $w$ returns a special feasible set, then so does a call to the oracle with objective $\overline{w}$.

\begin{claimInProof*}
For any set $Q\subseteq E$ with $|Q\cap S_i| =1$ for every $i\in [\sqrt{n}]$, we have:
\begin{equation*}
\text{If } \enspace w(Q) > \max_{U\in \mathcal{S}} w(U) \enspace, 
\text{ then } \enspace \overline{w}(Q) > \max_{U\in \mathcal{S}} \overline{w}(U)\enspace.
\end{equation*}
\end{claimInProof*}
\begin{proof}
We number the sets $S_i$ such that $w(S_1) \geq w(S_2) \geq \ldots \geq w(S_{\sqrt{n}})$. %
For $i\in [\sqrt{n}]$, let $q_i$ be the single element in $Q\cap S_i$. 
It then holds that
\begin{align*}
0 &< w(Q) - \max_{U\in \mathcal{S}} w(U)\\
  &= w(Q) - \sum_{i=1}^{\beta} w(S_i)\\
  &= \sum_{i=1}^{\beta} \big((w(q_i)-w(S_i)\big) + \sum_{i=\beta+1}^{\sqrt{n}} w(q_i)\\
  &= \sum_{i=1}^{\beta} w(S_i) \cdot \big(\overline{w}(q_i)-\overline{w}(S_i)\big)
       + \sum_{i=\beta+1}^{\sqrt{n}} w(S_i) \cdot \overline{w}(q_i)\\
  &\leq \sum_{i=1}^{\beta} w(S_{\beta}) \cdot \big(\overline{w}(q_i)-\overline{w}(S_i)\big)
       + \sum_{i=\beta+1}^{\sqrt{n}} w(S_{\beta}) \cdot \overline{w}(q_i)\\
  &= w(S_{\beta}) \cdot \Bigg( \sum_{i=1}^{\beta} \big(\overline{w}(q_i)-\overline{w}(S_i)\big)
       + \sum_{i=\beta+1}^{\sqrt{n}} \overline{w}(q_i) \Bigg)\\
  &= w(S_{\beta}) \cdot \Bigg(\overline{w}(Q) - \sum_{i=1}^{\beta} \overline{w}(S_i) \Bigg)\\
  &= w(S_{\beta}) \cdot \Big( \overline{w}(Q) - \max_{U\in \mathcal{S}} \overline{w}(U)\Big)\enspace,
\end{align*}
where the third equality follows from~\eqref{eq:relWToWBar}, and the non-strict inequality is a consequence of $w(S_1) \geq w(S_2) \geq \ldots \geq w(S_{\sqrt{n}})$.
Hence, we conclude that $\overline{w}(Q) - \max_{U\in \mathcal{S}} \overline{w}(U) > 0$, which proves the claim.
\end{proof}

The above claim implies that
\begin{equation*}
\Pr\Big[ w(T) > \max_{U\in \mathcal{S}}w(U)\Big] \leq
 \Pr\Big[ \overline{w}(T) > \max_{U\in \mathcal{S}} \overline{w}(U)\Big]\enspace,
\end{equation*}
and hence, to show Lemma~\ref{lem:singleOracleCallBad}, it suffices to prove
\begin{equation}\label{eq:enoughToBoundProbWBar}
\Pr\Big[\overline{w}(T) > \max_{U\in \mathcal{S}} \overline{w}(U)\Big] \leq n^{-d-1}\enspace.
\end{equation}
We first consider the value of $\max_{U\in \mathcal{S}} \overline{w}(U)$, where we recall that  $\overline{w}(S_i) \in \{0,1\}$ for each $i\in [\sqrt{n}]$. Defining $k\coloneqq |\{i\in [\sqrt{n}] \mid \overline{w}(S_i) = 1 \}|$, we thus have $\max_{U\in \mathcal{S}} \overline{w}(U) = \min\{k,\beta\}$. 
Observe that if $k\leq \beta$, then $k=\max_{U\in \mathcal{S}}\overline{w}(U) = \overline{w}(E)$, because in this case the union of all $S_i$ with $\overline{w}(S_i)=1$ is a set $U\in \mathcal{S}$ which contains all elements with non-zero $\overline{w}$-weight. Hence, if $k \leq \beta$, we can never have $\overline{w}(T) > \max_{U\in \mathcal{S}} \overline{w}(U)$, and~\eqref{eq:enoughToBoundProbWBar} trivially holds. Thus, we assume $k > \beta$, which implies
\begin{equation*}
\max_{U\in \mathcal{S}} \overline{w}(U) = \beta\enspace.
\end{equation*}
For $i\in [\sqrt{n}]$, let $t_i$ be the single element in $T\cap S_i$ and 
define the random variable
\begin{align*}
Y_i &\coloneqq \overline{w}(t_i) \enspace .
\end{align*}
Notice that the variables $Y_i$, $i \in [\sqrt{n}]$, are independent and take values within $[0,1]$. Moreover, because $t_i$ is chosen uniformly at random from $S_i$, we have
\begin{equation*}
\E[Y_i] = \frac{1}{\sqrt{n}}\cdot \overline{w}(S_i) \qquad \forall i\in [\sqrt{n}]\enspace.
\end{equation*}
We now define
\begin{align*}
Y &\coloneqq \sum_{i=1}^{\sqrt{n}} Y_i = \overline{w}(T)\enspace,
\end{align*}
and observe that by the above, it holds that
\begin{equation}\label{eq:expYLeq1}
\E[\overline{w}(T)] = \E[Y] = \E\Bigg[ \sum_{i=1}^{\sqrt{n}} Y_i \Bigg] = \sum_{i=1}^{\sqrt{n}} \frac{1}{\sqrt{n}} \cdot\overline{w}(S_i) \leq 1\enspace,
\end{equation}
where the last inequality follows from $\overline{w}(S_i) \leq 1$ for each $i\in [\sqrt{n}]$.

By a standard Chernoff bound (see, e.g.,~\cite{lehman_2017_mathematics}), %
we have that for any $\gamma \geq 1$, it holds
\begin{equation*}
\Pr[Y \geq \gamma \cdot \E[Y]]
 \leq e^{- \E[Y] \cdot(\gamma\log\gamma - \gamma +1)}
 \leq e^{- \E[Y] \cdot(\gamma\log\gamma - \gamma)}
 = \bigg(\frac{e}{\gamma}\bigg)^{ \gamma \cdot \E[Y]}\enspace.
\end{equation*}
Note that $\gamma \coloneqq \beta/\E[Y]$ satisfies $\gamma \geq \beta \geq 1$ because we have $\E[Y]\leq 1$ due to~\eqref{eq:expYLeq1}.
Hence, we can use the above Chernoff bound for this value of $\gamma $ to obtain
\begin{align*}
\Pr\Big[\overline{w}(T) > \max_{U\in \mathcal{S}}\overline{w}(U)\Big] &= \Pr[Y > \beta] \\
&\leq \Big(\frac{e}{\gamma}\Big)^{\beta} \\
&\leq \Big(\frac{e}{\beta}\Big)^{\beta} \\
&= e^{\beta \cdot (1-\log\beta)} \\
&\leq e^{c\cdot \frac{\log n}{\log\log n} \cdot (1 -\log c -\log\log n +\log\log\log n)}\\
&\leq n^{-\frac{c}{2}}\\
&= n^{-d-1}\enspace,
\end{align*}
where the second inequality follows from $\gamma \geq \beta$, and the fourth inequality holds for $n$ large enough; more precisely, for $n$ such that $\frac{1}{2}\log\log n \geq 1+ \log\log\log n$. The above shows~\eqref{eq:enoughToBoundProbWBar}, and thus completes the proof of Lemma~\ref{lem:singleOracleCallBad}.
\end{proof}

The difficulty in the example we provided to prove Theorem~\ref{thm:intro-hardness} stems from the constraints, and not from the submodular function. Thus, it would also be interesting to find an example with an explicitly given family of feasible sets which shows that even in this case, a linear optimization oracle does not suffice to get good approximations for submodular maximization.

\section{Conclusion}

In this work, we derived a general geometric condition for neighborhood functions, namely the property of being $\alpha$-conic, that allows for obtaining strong approximation factors for $\mathsf{CSFM}$ in both the monotone and non-monotone case via local search. Moreover, the property of being $\alpha$-conic is independent of the submodular function to be maximized.
The local search procedure we suggest, when used with an appropriately chosen neighborhood function, applies to  a large set of problems; in particular, it allows for replicating known approximation results for $k$-intersection systems and $k$-exchange systems in a unifying way and with a slightly improved running time.

Furthermore, we showed that being able to optimize linear functions over some constraints is, in general, not enough to efficiently find good approximations for submodular maximization over the same constraints. This is in stark contrast to the related problem of optimizing the multilinear extension of a submodular function, for which a linear optimization oracle suffices to obtain strong guarantees.

\bibliographystyle{plain}

\appendix

\section{Swap neighborhood function for $k$-intersection systems}\label{append:kIntersection}

In this section, we provide a proof of Theorem~\ref{thm:LS-geom-ksystems} for $k$-intersection systems, i.e., the intersection of $k$ matroids on a common ground set. 
Throughout this section, let $k\in \mathbb{Z}_{\geq 2}$ and $p\in \mathbb{Z}_{\geq 1}$, and let $M_i = (E,\mathcal{I}_i)$ for $i\in [k]$ be $k$ matroids defined on the same ground set $E$. We denote by $\mathcal{F} \coloneqq \bigcap_{i=1}^k \mathcal{I}_i$  the $k$-intersection system for which we want to show that $N_p^k$ is an $\alpha$-conic neighborhood function, %
 where $\alpha \coloneqq k-1+1/p$.
  Hence, we have to prove that for any $S,T\in \mathcal{F}$, it holds that 
\begin{equation}\label{eq:alphaSpanKInt}
\frac{1}{\alpha} \cdot \big( \chi^T + (\alpha-1)\cdot \chi^{S\cap T}\big)
 \in
C\coloneqq \Big( \chi^S + \cone\big(\{\chi^A - \chi^S \mid A\in N^k_p(S)\}\big) \Big)\enspace.
\end{equation}
To show~\eqref{eq:alphaSpanKInt}, we generally follow an approach used in~\cite{LeeSviridenkoVondrak2010}, combined with results from~\cite{chekuri_2011_multibudgeted}, which help us to provide a more streamlined proof. More precisely, as in~\cite{LeeSviridenkoVondrak2010}, we start by focusing on the first two matroids, and later consider the remaining $k-2$ ones.

We begin with a simplifying observation, namely that it suffices to prove~\eqref{eq:alphaSpanKInt} for $S,T\in \mathcal{F}$ being common bases of all $k$ matroids (and hence, in particular, they satisfy $|S|=|T|$).
\begin{lem}\label{lem:basesSufficesForSpanProp}
If~\eqref{eq:alphaSpanKInt} holds for any $k$ matroids on a common ground set and any common bases $\overline{S},\overline{T}$, %
 then it also holds for any $k$ matroids on a common ground set and any common independent sets $S,T$.
\end{lem}
\begin{proof}
	Let $M_i = (E,\mathcal{I}_i)$ for $i\in [k]$ be the $k$ matroids defined on a common ground set $E$, and let $S,T \in \mathcal{F} = \bigcap_{i=1}^k \mathcal{I}_i$ be the common independent sets for which we want to prove~\eqref{eq:alphaSpanKInt}.
We add a set $W$ of $|E|$ new dummy elements to all matroids $M_i$, $i\in [k]$, and truncate them at cardinality $|E|$. More precisely, each matroid $M_i=(E,\mathcal{I}_i)$, $i\in [k]$, gets extended to a matroid $\overline{M}_i=(E\dcup W, \overline{\mathcal{I}}_i)$, where $\overline{\mathcal{I}}_i=\{\overline{I} \subseteq E\dcup W \mid |\overline{I}|\leq |E|,\, \overline{I}\cap E \in \mathcal{I}_i\}$.
 Notice that for each $i\in [k]$, any basis in $\overline{M}_i$ has cardinality $|E|$.
 Let $\overline{\mathcal{F}} = \bigcap_{i=1}^k \overline{\mathcal{I}_i}$ be the set of common independent sets of all extended matroids $\overline{M}_i$, $i\in [k]$, and let $\overline{\mathcal{B}} = \{\overline{I} \in \overline{\mathcal{F}} \mid |\overline{I}|=|E|\}$ be the set of common bases.
By adding $|E|-|S|$ arbitrary elements of $W$ to $S$, and analogously $|E|-|T|$ arbitrary elements of $W$ to $T$, we get two sets $\overline{S}\supseteq S, \overline{T}\supseteq T$ with $\overline{S},\overline{T}\in \overline{\mathcal{B}}$. 
If~\eqref{eq:alphaSpanKInt} holds for common bases of $k$ arbitrary matroids on a common ground set, we can  apply~\eqref{eq:alphaSpanKInt} to the common bases $\overline{S}$ and $\overline{T}$ of the extended matroids. Observe that, by only considering the coordinates corresponding to $E$, we obtain the desired relation~\eqref{eq:alphaSpanKInt} for the sets $S$ and $T$. 
Here, it is important to note that if $\overline{U}\subseteq E\dcup W$ is in the neighborhood $N_p^k(\overline{S})$ of $\overline{S}$ with respect to $\overline{\mathcal{F}}$, then $U \coloneqq \overline{U}\cap E$ is in the neighborhood $N_p^k(S)$ of $S$ with respect to $\mathcal{F}$.
\end{proof}
Hence, let $\mathcal{B} \subseteq \mathcal{F}$ be the set of all common bases of $M_1, \ldots, M_k$.
We will show~\eqref{eq:alphaSpanKInt} only for $S,T\in \mathcal{B}$.
By Lemma~\ref{lem:basesSufficesForSpanProp}, this suffices to prove Theorem~\ref{thm:LS-geom-ksystems} for $k$-intersection systems.

The following result is a key ingredient to deal with the first two matroids $M_1$ and $M_2$. It is a slight rephrasing of Lemma~3.3 in~\cite{chekuri_2011_multibudgeted}, where we additionally dropped some of the properties we do not need.
\begin{lem}[{Lemma~3.3 in~\cite{chekuri_2011_multibudgeted}}]\label{lem:2matIntersection}%
\footnote{The sets $P_j$ in~\cite{chekuri_2011_multibudgeted} are (the vertex sets of) so-called irreducible cycles/paths in the exchange digraph $D_{M_1,M_2}(S)$ of $M_1$ and $M_2$ using only elements of $S \sdiff T$, which are sets fulfilling property~\ref{item:symDiffIndep} of Lemma~\ref{lem:2matIntersection} and satisfying $\big| |P_j\cap S| - |P_j\cap T|\big| \leq 1$. The latter means that point~\ref{item:PjAreSmall} of our lemma is implied by Lemma~3.3 in~\cite{chekuri_2011_multibudgeted}, which states that $|P_j|\leq 2p+1$ (in the lemma, the length of an irreducible cycle/path refers to the number of its vertices);
	indeed, $S$ and $T$ being common bases implies $|P_j \sdiff S|\leq |S|$, as $P_j \sdiff S$ is independent in $M_1$ and $M_2$, which can be rephrased as $|P_j \cap S| \geq |P_j \cap T|$.
	 Another minor difference is that Lemma~3.3 in~\cite{chekuri_2011_multibudgeted} states point~\ref{item:goodSymDiffCoverage} using coefficients, i.e., $\sum_{j=1}^m \lambda_j \chi^{P_j} = (1+\frac{1}{p})\cdot \lambda \cdot \chi^{S\setminus T} + \lambda \cdot \chi^{T\setminus S}$, where $\lambda_j\geq 0$ for $j\in [m]$. Since the coefficients $\lambda_j$ can be chosen to be rational (they are a solution to a linear equation system with only rational entries), one can first scale up both sides of the equation to obtain integer coefficients, and then add to the collection additional copies of the sets $P_j$ to obtain coefficients equal to $1$ as stated in point~\ref{item:goodSymDiffCoverage} of Lemma~\ref{lem:2matIntersection}.
}
Let $M_i=(E,\mathcal{I}_i)$ for $i =1,2$ be two matroids on a common ground set $E$, let $S,T$ be common bases of $M_1$ and $M_2$, and let $p\in \mathbb{Z}_{\geq 1}$. Then, there exist sets $P_1,\ldots, P_m\subseteq S\sdiff T$ and an integer $r \in \mathbb{Z}_{\geq 1}$ such that:
\begin{enumerate}[label=\normalfont(\roman*),itemsep=0em,topsep=0.2em]
\item\label{item:symDiffIndep} $P_j \sdiff S \in \mathcal{I}_1\cap \mathcal{I}_2$ for each $j\in [m]$,
\item\label{item:PjAreSmall} $|P_j \cap S| \leq p+1$ and $|P_j \cap T|\leq p$ for each $j\in [m]$,
\item\label{item:goodSymDiffCoverage} $\sum_{j=1}^m \chi^{P_j} = \big(1+\frac{1}{p}\big)\cdot r \cdot \chi^{S\setminus T} + r\cdot \chi^{T\setminus S}$.
\end{enumerate}
\end{lem}

Moreover, to deal with the matroids $M_i$ for $i=3,\ldots, k$, we use the following lemma from~\cite{LeeSviridenkoVondrak2010}, which we will slightly modify afterwards.

\begin{lem}[Lemma~2.7 in~\cite{LeeSviridenkoVondrak2010}]\label{lem:kIntersection-auxiliary2a}
	Let $M = (E, \mathcal{I})$ be a matroid and $I,J \in \mathcal{I}$. Let $I_1, \ldots, I_m \subseteq I$ such that each element of $I$ appears in at most $q$ of them. Then, there are $J_1, \ldots, J_m \subseteq J$ such that each element of $J$ appears in at most $q$ of them and, for each $i \in [m]$, it holds that $I_i \cup (J \setminus J_i ) \in \mathcal{I}$.
\end{lem}
The following slight modification of the above lemma is suitable for our purposes.
\begin{lem}\label{lem:kIntersection-auxiliary2b}
	Let $M = (E, \mathcal{I})$ be a matroid and $I,J \in \mathcal{I}$. Let $I_1, \ldots, I_m \subseteq I\setminus J$ such that each element of $I \setminus J$ appears in at most $q$ of them. Then, there are  $J_1, \ldots, J_m \subseteq J \setminus I$ such that each element of $J \setminus I$ appears in at most $q$ of them and, for each $i \in [m]$, it holds that $|J_i| \leq |I_i|$ and $I_i \cup (J \setminus J_i ) \in \mathcal{I}$.
\end{lem}
\begin{proof}
	We consider the contraction $M / (I \cap J) $ of $M$ onto $E \setminus (I \cap J)$, where we denote the resulting matroid by $M' = (E \setminus (I \cap J), \mathcal{I}')$.
	Letting $I' \coloneqq I \setminus J$ and $J' \coloneqq J \setminus I$, it clearly holds that $I', J' \in \mathcal{I}'$.
	Moreover, $I_1, \ldots, I_m \subseteq I'$ are such that each element of $I'$ appears in at most $q$ of them.
	Thus, by Lemma~\ref{lem:kIntersection-auxiliary2a} applied to $M'$, there are $J_1, \ldots, J_m \subseteq J'$ such that each element of $J'$ appears in at most $q$ of them and, for each $i \in [m]$, it holds that $I_i \cup (J' \setminus J_i) \in \mathcal{I}'$.
	
	For every $i \in [m]$, we (greedily) remove as many elements from $J_i$ as we can without violating the second property above. 
	Clearly, this does not affect the first property.
	Assume that after this, we have $|J_i | > |I_i|$ for some $i \in [m]$.
	This means that $| I_i \cup (J' \setminus J_i) | = |I_i| + |J'| - |J_i| < |J'|$.
	Since $I_i \cup (J' \setminus J_i ) \in \mathcal{I}'$ and $J' \in \mathcal{I}'$, there is an element $e \in J' \setminus (I_i \cup (J' \setminus J_i)) = J_i$ such that $(I_i \cup (J' \setminus J_i)) \cup \{e\} \in \mathcal{I}'$.
	However, this means that $I_i \cup (J' \setminus (J_i \setminus \{e\})) = (I_i \cup (J' \setminus J_i)) \cup \{e\} \in \mathcal{I}'$.
	Since $e \in J_i$, this contradicts the minimality of $J_i$. 
	Hence, it holds for all $i \in [m]$ that $|J_i| \leq |I_i|$.
	
	To conclude, we note that $J_1, \ldots, J_m \subseteq J' = J \setminus I$ and each element of $J' = J \setminus I$ appears in at most $q$ of them.
	Moreover, it holds for each $i \in [m]$ that $|J_i | \leq |I_i|$ and $I_i \cup (J' \setminus J_i) \in \mathcal{I}'$.
	The latter implies that
	$ I_i \cup (J \setminus J_i) = (I_i \cup (J' \setminus J_i)) \cup (I \cap J) \in \mathcal{I}    $
	for every $i \in [m]$, where we used the definition of $\mathcal{I}'$ and that $I \cap J \in \mathcal{I}$.
	This finishes the proof.	
\end{proof}

Finally, we exploit the following basic property of the shifted cone $C$, which will also be useful when we talk about $k$-exchange systems in Appendix~\ref{append:kExchange}.
\begin{lem}\label{lem:CDownClosed}
Let $(E,\mathcal{I})$ be an independence system,\footnote{An independence system $(E,\mathcal{I})$ consists of a finite ground set $E$ and a non-empty family $\mathcal{I}\subseteq 2^E$ of subsets of $E$ that is closed under taking subsets, i.e., if $I\in \mathcal{I}$ and $J\subseteq I$, then $J\in \mathcal{I}$.} %
 and for $k,p\in \mathbb{Z}_{\geq 1}$, let $N_p^k$ be the $(k,p)$-swap neighborhood function for $\mathcal{I}$. Then, for any $S\in \mathcal{I}$, the set
\begin{equation*}
C\coloneqq \Big( \chi^S + \cone\big(\{\chi^A - \chi^S \mid A\in N^k_p(S)\}\big) \Big)
\end{equation*}
satisfies that if $z \in C$ and $0\leq y \leq z$, then $y\in C$.
\end{lem}
\begin{proof}
Let $y$ be a point such that there exists $z \in C$ with $0 \leq y \leq z$.	
Among all points $x'\in C$ that satisfy $x'\geq y$, let $x$ be one minimizing $\|x'-y\|_1$. Such a minimizer exists because $C$ is closed. We claim that $x=y$, which implies the lemma. 
First, observe that $\cone\big(\{\chi^A - \chi^S \mid A\in N_p^k(S)\}\big)$ contains all directions $-\chi^{\{e\}}$ for any $e\in S$. This holds because $\mathcal{I}$ is an independence system, which guarantees $S \setminus \{e\} \in \mathcal{I}$, and thus $S\setminus \{e\}\in N_p^k(S)$. %
Hence, for any point in $C$, we can arbitrarily decrease coordinates corresponding to elements in $S$ and again obtain a point in $C$. Thus, we must have $x(e) = y(e)$ for all $e\in S$. Moreover, for any $e\in E\setminus S$, we claim that $x(e) > y(e)$ leads to a contradiction.  Indeed, let $\lambda_A \geq 0$ for $A\in N_p^k(S)$ such that
\begin{equation*}
x = \chi^S + \sum_{A\in N_p^k(S)} \lambda_A  (\chi^A - \chi^S)\enspace.
\end{equation*}
Such coefficients exist because $x\in C$. We now consider $e\in E\setminus S$. Notice that for any $A\in N_p^k(S)$, we have $A\setminus \{e\} \in N_p^k(S)$ because $\mathcal{I}$ is an independence system and $e \in E \setminus S$. Hence, it holds that
\begin{equation*}
w \coloneqq \chi^S + \sum_{A \in N_p^k(S)} \lambda_A  (\chi^{A\setminus \{e\}} - \chi^S)
\end{equation*}
satisfies $w\in C$ as well. Moreover, $w$ is identical to $x$ except (potentially) in the coordinate corresponding to $e$, where $w(e) =0$. Therefore, if $x(e) > y(e)$, there is a convex combination $\overline{x}$ of $x\in C$ and $w\in C$ %
such that $\overline{x}$ is identical to $x$ except for the coordinate corresponding to $e$, where we have $\overline{x}(e) = y(e)$.
Since $C$ is convex, it holds that $\overline{x} \in C$, which violates the assumption that $x$ minimizes $\|x'-y\|_1$ among all $x' \in C$ with $x'\geq y$.
\end{proof}

With this, we are now ready to prove Theorem~\ref{thm:LS-geom-ksystems} for $k$-intersection systems.
\begin{proof}[Proof of Theorem~\ref{thm:LS-geom-ksystems} for $k$-intersection systems]
Due to Lemma~\ref{lem:basesSufficesForSpanProp}, we only need to consider two common bases $S,T$ of the matroids $M_1, \ldots, M_k$, and show that~\eqref{eq:alphaSpanKInt} holds for them.

Let  $P_1, \ldots, P_m \subseteq S\sdiff T$  %
and $r \in \mathbb{Z}_{\geq 1}$ be
as guaranteed by Lemma~\ref{lem:2matIntersection} for the two matroids $M_1$ and $M_2$.
Setting $Q_j \coloneqq P_j \sdiff S$ for $j\in [m]$, we obtain by point~\ref{item:goodSymDiffCoverage} of Lemma~\ref{lem:2matIntersection} that
\begin{equation}\label{eq:shiftWithQj}
\begin{aligned}
\sum_{j=1}^m \big(\chi^{Q_j} - \chi^S \big) &= \sum_{j=1}^m \big(\chi^{P_j \sdiff S} - \chi^S \big) \\
  &= \sum_{j=1}^m \big(- \chi^{P_j\cap S} + \chi^{P_j \cap T}\big)\\
  &= - \Big(1+\frac{1}{p}\Big) \cdot r \cdot \chi^{S\setminus T} + r \cdot \chi^{T\setminus S}\enspace.
\end{aligned}
\end{equation}

Due to point~\ref{item:symDiffIndep} of Lemma~\ref{lem:2matIntersection}, 
the sets $Q_j$, $j \in [m]$, are feasible in $M_1$ and $M_2$. Moreover, to turn them into feasible sets for all $k$ matroids, we use Lemma~\ref{lem:kIntersection-auxiliary2b} to drop some elements.
By point~\ref{item:goodSymDiffCoverage} of Lemma~\ref{lem:2matIntersection}, we know that each element of $T\setminus S$ appears in precisely $r$ sets among $P_1\setminus S, \ldots, P_m \setminus S$. Hence, we can apply Lemma~\ref{lem:kIntersection-auxiliary2b} to the matroids $M_i$ %
to obtain sets $J_{ij}\subseteq S\setminus T$ for $i\in \{3,\ldots,k\}$ and $j\in [m]$ such that:
\begin{enumerate}[label=(\roman*),itemsep=0.0em,topsep=0.2em]
\item\label{item:remFeasible} $(Q_{j}\setminus J_{ij}) \subseteq ((P_j \cup S) \setminus J_{ij}) \in \mathcal{I}_i \quad \forall i\in \{3,\ldots,k\}, j\in [m]$,
\item $\sum_{j=1}^m \chi^{J_{ij}} \leq r\cdot \chi^{S\setminus T} \quad \forall i\in \{3,\ldots,k\}$,
\item $|J_{ij}|\leq |P_j\setminus S| \leq p \quad \forall i\in \{3,\ldots,k\}, j\in [m]$.
\end{enumerate}

We now define for every $j \in [m]$ the set
\begin{equation*}
A_j \coloneqq Q_j \setminus \Big(\bigcup_{i=3}^k J_{ij}\Big)  \enspace.
\end{equation*}
First, notice that for each $j \in [m]$, it holds that $A_j \in N_p^k(S)$ because $A_j\in \mathcal{F} = \bigcap_{i=1}^k \mathcal{I}_i$ due to~\ref{item:remFeasible} above, $|A_j\setminus S| = |Q_j\setminus S| = |P_j\cap T| \leq p$, and $|S\setminus A_j| \leq |P_j\cap S| + \sum_{i=3}^k|J_{ij}| \leq p+1 + (k-2) p$. Moreover, for $\alpha= k-1+1/p$, we get
\begin{align*}
\chi^S + \frac{1}{r \alpha} \cdot \sum_{j=1}^m \big(\chi^{A_j} - \chi^S \big)
  &\geq
\chi^S + \frac{1}{r \alpha} \cdot
  \Bigg( \sum_{j=1}^m \big(\chi^{Q_j} - \chi^S \big)
 -\sum_{j=1}^m \sum_{i=3}^k \chi^{J_{ij}}\Bigg)\\
  &\geq \chi^S + \frac{1}{r \alpha} \cdot \Big( - \Big(1+\frac{1}{p}\Big)\cdot r - (k-2)\cdot r \Big) \cdot \chi^{S\setminus T} + \frac{1}{\alpha} \cdot \chi^{T\setminus S}\\
  &= \chi^S - \chi^{S\setminus T} + \frac{1}{\alpha} \cdot \chi^{T\setminus S}\\
  &= \frac{1}{\alpha} \cdot \big(\chi^{T} + (\alpha-1)\cdot \chi^{S\cap T} \big)\enspace,
\end{align*}
where the second inequality uses~\eqref{eq:shiftWithQj} and the fact that $\sum_{j=1}^m \chi^{J_{ij}}  \leq r\cdot \chi^{S\setminus T}$ for every $i\in \{3,\ldots,k\}$.
Hence, the above shows that there is a point $z\in C = \big( \chi^S + \cone\big(\{\chi^A - \chi^S \mid A\in N_p^k(S)\}\big)\big)$ such that $0\leq \frac{1}{\alpha} \cdot (\chi^T + (\alpha-1) \cdot \chi^{S\cap T}) \leq z$, where $z$ can actually be chosen as the left-hand side in the above inequality. By Lemma~\ref{lem:CDownClosed}, this implies that $\frac{1}{\alpha}\cdot (\chi^T + (\alpha-1)\cdot \chi^{S\cap T}) \in C$, as desired.
\end{proof}

\section{Swap neighborhood function for $k$-exchange systems}\label{append:kExchange}

In this section, we prove the statement of Theorem~\ref{thm:LS-geom-ksystems} for $k$-exchange systems. 
Together with Appendix~\ref{append:kIntersection}, in which we showed the statement for $k$-intersection systems, this completes the proof of Theorem~\ref{thm:LS-geom-ksystems}.
Our proof for $k$-exchange systems closely follows an approach presented in~\cite{FeldmanNaorSchwartzWard2011}.

We start with a formal definition of $k$-exchange systems.
\begin{defn}[Definition~1 in~\cite{FeldmanNaorSchwartzWard2011}]\label{defn:kExchange}
	An independence system %
	 $(E, \mathcal{F})$ is a \emph{$k$-exchange system} if, for all $S$ and $T$ in $\mathcal{F}$, there exists a multiset $Y = \{ Y_e \subseteq S \setminus T \mid e \in T \setminus S  \}$ such that:
	\begin{enumerate}[label=\normalfont(\roman*),itemsep=-0.2em,topsep=0em]
		\item $|Y_e| \leq k$ for each $e \in T \setminus S$,
		\item every $e' \in S \setminus T$ appears in at most $k$ sets of $Y$,
		\item for all $T' \subseteq T \setminus S$, $\big(  S \setminus ( \bigcup_{e \in T'} Y_e  )  \big)\cup T' \in \mathcal{F}$.
	\end{enumerate}
\end{defn}

In~\cite{FeldmanNaorSchwartzWard2011}, useful properties regarding the neighborhood function $N_p^k$ for $k$-exchange systems are provided. The lemma we state below comprises several separate statements in~\cite{FeldmanNaorSchwartzWard2011}, which we combined for convenience. More precisely, points~\ref{item:kExCoverageT} and~\ref{item:kExCoverageS} correspond to Lemma~1 in~\cite{FeldmanNaorSchwartzWard2011}. Moreover, points~\ref{item:kExFeasible} and~\ref{item:kExSizes} are immediate from the construction of $\mathcal{P}$ (which is called $\mathcal{P}'$ in~\cite{FeldmanNaorSchwartzWard2011}); these properties are also explicitly stated at the beginning of the proof of Theorem~6 in~\cite{FeldmanNaorSchwartzWard2011}.
\begin{lem}[see Section~4.1 in~\cite{FeldmanNaorSchwartzWard2011}]\label{lem:goodExchangesKEx}
Let $k\in \mathbb{Z}_{\geq 2}$ and $p\in \mathbb{Z}_{\geq 1}$.
Moreover,  let $(E,\mathcal{F})$ be a $k$-exchange system, and let $S,T\in \mathcal{F}$.  
Then, there is a multiset $\mathcal{P}$ containing subsets of $S\sdiff T$ and an integer $n\in \mathbb{Z}_{\geq 1}$ (depending on $k$ and $p$) such that:
\begin{enumerate}[label=\normalfont(\roman*),itemsep=-0.2em,topsep=0.2em]
\item\label{item:kExFeasible} $P \sdiff S \in \mathcal{F}$ for all $P\in \mathcal{P}$,
\item\label{item:kExCoverageT} each element $e\in T\setminus S$ appears in precisely $2pn$ sets of $\mathcal{P}$,
\item\label{item:kExCoverageS} each element $e'\in S\setminus T$ appears in at most $2((k-1)p + 1)n$ sets of $\mathcal{P}$,
\item\label{item:kExSizes} $|P\setminus S|\leq p$ and $|P \cap S|\leq (k-1)p+1$ for all $P\in \mathcal{P}$.
\end{enumerate}
\end{lem}

The above lemma now enables us to prove Theorem~\ref{thm:LS-geom-ksystems} for $k$-exchange systems.
\begin{proof}[Proof of Theorem~\ref{thm:LS-geom-ksystems} for $k$-exchange systems]
Let $\mathcal{F}$ be a $k$-exchange system, where $k \in \mathbb{Z}_{\geq 2}$, and let $S,T\in \mathcal{F}$.
Moreover, for $p \in \mathbb{Z}_{\geq 1}$, let $\mathcal{P}=\{P_1,\ldots, P_m\}$ and $n \in \mathbb{Z}_{\geq 1}$  be %
as guaranteed by Lemma~\ref{lem:goodExchangesKEx}. %
For all $j \in [m]$, we now define
\begin{equation*}
A_j \coloneqq P_j \sdiff S \enspace.
\end{equation*}
By points~\ref{item:kExFeasible} and~\ref{item:kExSizes} of Lemma~\ref{lem:goodExchangesKEx}, we have $A_j \in N_p^k(S)$ for every $j\in [m]$. Moreover, we get from points~\ref{item:kExCoverageT} and~\ref{item:kExCoverageS} of Lemma~\ref{lem:goodExchangesKEx} that
\begin{equation*}
\sum_{j=1}^m \big(\chi^{A_j} -\chi^{S}\big) \geq 2pn \cdot \chi^{T\setminus S} - 2\alpha p n\cdot\chi^{S\setminus T}\enspace,
\end{equation*}
where $\alpha\coloneqq k-1 + 1/p$.  
The above inequality implies
\begin{align*}
\chi^S + \frac{1}{2\alpha p n} \cdot \sum_{j=1}^m \big(\chi^{A_j} - \chi^S\big)
   &\geq \chi^S + \frac{1}{\alpha}\cdot \chi^{T\setminus S} - \chi^{S\setminus T} \\
   &= \frac{1}{\alpha}\cdot\big(\chi^T + (\alpha-1)\cdot \chi^{S\cap T}\big)\enspace.
\end{align*}
Analogous to the discussion in Appendix~\ref{append:kIntersection} for $k$-intersection systems, this shows that the point
\begin{equation*}
z\coloneqq \chi^S + \frac{1}{2\alpha p n} \cdot \sum_{j=1}^m \big(\chi^{A_j} - \chi^S\big)
\end{equation*}
satisfies $0\leq \frac{1}{\alpha} \cdot (\chi^T + (\alpha-1) \cdot \chi^{S\cap T}) \leq z$.
 Since $z\in C \coloneqq \big( \chi^S + \cone\big(\{\chi^A - \chi^S \mid A\in N_p^k(S)\}\big)\big)$, Lemma~\ref{lem:CDownClosed} yields that $\frac{1}{\alpha} \cdot (\chi^T + (\alpha-1) \cdot \chi^{S\cap T})\in C$, as desired.
\end{proof}

\end{document}